\documentclass[12pt,a4paper, left=0in, right=0in, top=0in, bottom=0in, margin=0.5in]{article}
\usepackage{geometry} 

\usepackage{titlesec}
\usepackage{setspace}
\usepackage{lipsum}
\usepackage{booktabs}
\usepackage{natbib}
\usepackage[ruled]{algorithm2e}
\usepackage{amssymb}
\usepackage{amsmath}
\usepackage{amsthm}
\usepackage{array}
\usepackage{float}
\usepackage{mdframed}
\usepackage{appendix}
\usepackage{multirow}
\usepackage{graphicx}
\usepackage{soul}
\usepackage{booktabs}
\usepackage[hyperindex,breaklinks]{hyperref}
\hypersetup{colorlinks = true, urlcolor = green, linkcolor = blue, citecolor = blue}
\usepackage[bottom]{footmisc}
\usepackage{hyperref}

\usepackage{eso-pic}
\usepackage{color}

\newtheorem{definition}{Definition}
\newtheorem{example}{Example}
\newtheorem{assumption}{Assumption}
\newtheorem{hypothesis}{Hypothesis}
\newtheorem{theorem}{Theorem}
\newtheorem{lemma}{Lemma}
\newtheorem{proposition}{Proposition}
\newtheorem{remark}{Remark}

\theoremstyle{definition}
\newtheorem{proced}{Procedure}

  {\par\noindent\hspace*{0.5em}\begin{minipage}{\textwidth}\ignorespaces}%
  {\end{minipage}\par}

\title{\bfseries Randomization tests for model specification\\ in causal inference under network interference}
\author{Supriya Tiwari\footnote{Email: supriya\_tiwari@isb.edu \\ Pallavi Basu's research is partially supported by the SERB MATRICS award MTR/2022/000073. The author(s) would like to thank the Isaac Newton Institute for Mathematical Sciences, Cambridge, for support and hospitality during the program Causal inference: From theory to practice and back again, where part of the work on this paper was undertaken. This work was supported by EPSRC grant no EP/Z000580/1.}\hspace{0.3cm}and\hspace{0.15cm}Pallavi Basu\\
  \selectfont Indian School of Business, Hyderabad, 500111, India}
\date{\today}

\begin{document}

\maketitle

\abstract{\noindent
Analysis of experimental data becomes challenging when the underlying population is connected by a network. Exposure mapping is a common tool in the literature for defining and estimating spillover effects. These mappings reduce the dimensionality of the estimand, thereby facilitating identifiability. It is assumed that this mapping is correctly specified, leaving the choice of the exposure mapping to the analyst. This makes estimators of the spillover effect, such as the Horvitz-Thompson estimator, vulnerable to bias from model misspecification. Although these estimators have been shown to be robust to certain forms of controlled misspecification, there has been relatively little methodological progress in empirically investigating appropriate exposure mappings. In this paper, we propose a novel design-based model specification framework for causal inference. Building on this, we develop a randomization-testing procedure to assess the correct specification of an exposure-mapping model in the presence of network interference. We provide theoretical guarantees for the asymptotic validity of the proposed testing procedure. We establish the favorable power properties of our method through an extensive simulation study and illustrate it in a field experiment investigating the effect of anti-conflict norms among adolescents.  

\begin{quote}
\end{quote}

\section{Introduction}
Randomized experiments are the cornerstone of causal inference. A key assumption in the estimation of treatment effects is the Stable Unit Treatment Value Assumption (SUTVA). This rules out the possibility of interference in the population of interest: the treatment status of one unit does not affect the potential outcomes of other units. This assumption is violated in many settings, and there is increasing interest in studying spillover effects arising from interference structured by a network (e.g., \cite{paluck2016changing}; \cite{cai2015social}). Estimating spillover effects is impossible without making any structural assumptions, since the potential outcomes depend on the treatment status of all units. The number of counterfactuals in this case grows exponentially in the size of total units, and identifying causal estimands is challenging even under restricted forms of interference (\cite{halloran1995causal}; \cite{sobel2006randomized}). A growing body of literature tackles this issue by using \textit{exposure mappings}. Proposed in \cite{hong2006evaluating} and formalized in \cite{manski2013identification} and \cite{aronow2017estimating}, an exposure mapping restricts the cardinality of interference parameters in the potential outcome function to enable the identification of spillover effects and can be interpreted as a cumulative treatment. This lower-dimensional representation not only helps estimate the spillover effect but also defines spillover estimands (c.f. \cite{savje2024causal}). Assuming that the specified exposure mapping captures all causal interference mechanisms, one can frame potential-outcome counterfactuals sliced at each exposure level. This generalizes the SUTVA assumption to a setting with potential interference at the exposure-mapping level (\cite{savje2021average}), thereby reducing the complexity of the interference problem. \\
Existing methods assume that these exposure mappings are correctly specified, that is, they capture all the relevant spillover information in the population of interest. Some recent works have investigated the impact of misspecification in exposure mapping on the performance of standard estimators (e.g., \cite{leung2022causal}; \cite{savje2024causal}). A central finding across these works is the robustness of standard estimators to some controlled forms of misspecification. Another framework proposed in the literature to ensure robustness to interference is the use of marginalized estimands for the spillover effect (e.g., \cite{savje2021average}; \cite{hu2022average}; \cite{wang2020design}). This helps study spillover effects without the explicit use of exposure mappings and stands robust to many forms of interference. One drawback is that no insight into the spillover mechanism can be gained because the estimand is marginalized. \cite{savje2024causal} delineates the use of exposure mapping, serving two distinct roles: defining the spillover estimand and encoding the causal mechanism in the study. While the spillover estimand can, in principle, be defined using any exposure mapping of interest, it is in how exposure mapping captures the causal mechanism that misspecification plays a role.\\
In this paper, we focus on developing a data-driven method to test if the exposure mapping is correctly specified in an experimental study. This helps researchers empirically validate the underlying causal mechanism. This can also inform the researcher in defining the spillover causal estimands, helping with the interpretability of the results. We work under the design framework that leverages the treatment assignment mechanism for estimation and inference. This contrasts with the super-population perspective, where the sampling process drives inference. Alternatively, one can view the design framework as conditional on the sampled units. Our first contribution is to develop a framework for testing the correct model specification of exposure mapping, formulated as a \textit{sharp} null hypothesis. In a concurrent work, \cite{gao2026impossibility} also develops a framework for specification testing under the sharp null and presents a novel impossibility result: without imposing restrictions on the model space of potential outcome functions, it is impossible to construct an informative randomization testing procedure for specification testing of the exposure mapping model. They then present a specification testing procedure with the alternative model space restricted to a linear-in-means model, and demonstrate its consistency. A key contribution of \cite{gao2026impossibility} is the presentation of theoretical power analysis and consistency for the proposed randomization testing procedure. We work with a broader class of potential outcome model functions and illustrate the power of our analysis through numerical studies.

The sharp null can be assessed using the standard `Fisher Randomization Test' (FRT). There is a burgeoning set of literature that proposes the use of randomization testing procedures for spillover effect (\cite{athey2018exact}; \cite{basse2019randomization}; \cite{puelz2022graph}; \cite{zhang2025multiple}; \cite{zhong2024unconditional}; \cite{tiwari2024quasi}). Refer to \cite{zhang2023randomization} for a detailed discussion on randomization tests. \cite{pouget2019testing} develops a hierarchical experimental design approach that combines a completely randomized trial and a cluster-randomized trial to test for interference, in contrast to Fisherian-style testing approaches. These methods test for the presence of interference and do not address the exact mechanism by which spillover effects occur. \cite{hoshino2026conditional} considers a hierarchical relationship between exposure mapping and alternative finer exposure mapping to build an exposure mapping specification test. The framework requires specifying the (finer) alternative exposure mapping, with the choice of this finer exposure mapping to be defined by the practitioner. In our framework, we define a finite-sample residual and develop a randomization testing procedure to test for a correctly specified interference mechanism. A key theme among the conditional randomization testing literature is ensuring the imputability of the potential outcome and constructing a valid test. Our work differs from this setup in that we do not observe the true residuals; instead, we use an estimate in the testing procedure. Our second contribution is to prove theoretical guarantees for the asymptotic validity of the estimated randomization test within the design framework. \cite{abadie2020sampling} reconciles standard errors in the sampling-based perspective with the design perspective for OLS. \cite{toulis2025asymptotic} provides theoretical bounds for the size of discrepancy between the approximate randomization test and the oracle randomization test at a given sample size. In this work, we work with a design interpretation of the OLS estimator under dependence induced by the interference structure. We show that the OLS coefficient of the outcome regressed on a given function transform of the exposure mapping converges to a constant as the number of observations increases if the exposure mapping is correctly specified. This asymptotic invariance to the treatment assignment vector under the null aids in building a randomization testing procedure, similar to the imputation of the potential outcome function under the null in FRT. We also show that this approximate randomization test is asymptotically valid in the presence of interference. For the testing procedure, we propose a novel test statistic based on the graphical correlation between the treatment and the residuals. \\ 
We present the outline of the paper. In Section \ref{section: setup}, we present notation and formalize the model specification testing framework. Section \ref{section: main result} details the main causal specification test hypothesis testing procedure, and introduces the test statistic proposed. Section \ref{section: est test} provides the theoretical guarantees for the asymptotic validity of the approximate randomization test. In Section \ref{section: sim study}, we conduct an extensive simulation study to numerically validate the results and demonstrate the method's power. Section \ref{section: application} illustrates our methodology on an anti-conflict intervention experimental dataset from \cite{paluck2016changing}, and Section \ref{section: conclusion} concludes the paper with future research avenues. All proofs and additional results are provided in the Appendix. 

\section{Setup}
\label{section: setup}
\subsection{Notation}
Consider a population $\mathbb{P} = \{1,2,...N\}$ of size $N$. The population units can be characterized by a set of covariate vectors, denoted by $X$. The population is also connected via a network, which we denote by the adjacency matrix with $G \in \{0,1\}^{N \times N}$. Here, $G_{ij} = 1$ indicates that units $i$ and $j$ are connected, consistent with the adjacency matrix. We assume $G$ is a simple undirected graph. We define $N_i$ as all units that are connected to $i$. That is, $N_i := \{j: G_{ij} = 1\}$.  We define the set $\bar{N_i} := N_i \bigcup \{i\}$. We denote the distance between units $i$ and $j$ by $dist(i,j)$. It represents the shortest path length between the units $i$ and $j$. By convention, we take $dist(i, i) = 0$, and $dist(i,j) = \infty$ if there does not exist any path between the units $i$ and $j$ for any two units $i$ and $j$ in $\mathbb{P}$. We denote $\mathcal{N}^{(d)}(i) := \{j: dist(i,j) = d\}$. We restrict our attention to a binary treatment assignment vector and denote it by $Z \in \{0,1\}^N$. Here, $Z_i=1$ corresponds to a treated unit, and $Z_i = 0$ corresponds to a control unit. For any $c \subseteq [N]$, we define $Z_c:= (Z_i : i \in c)$. The corresponding potential outcome function is represented by $Y_{X, G}(Z)$, where the covariates and the network remain fixed. For readability, we will use the notation $Y(Z)$ in the rest of the article unless deemed necessary. An exposure mapping $e_G:\{0,1\}^N \rightarrow \mathbb{R}^{N \times (k+1)}$ is a function of the treatment assignment vector mapping such that ${e_G}_i = (f_{ij}(Z, G))_{j\in[k+1]}$. Here, $f_{ij}(.)$ is a real-valued function such that $f_{i1} = Z_i$, and $f_{ij}(Z,G) = f(Z_{\mathcal{N}^{(j-1)}(i)},G)$ for $1< j\leq k+1$ with $k$ fixed. We treat the network as fixed, and use the notation $e$ for exposure mapping. The probability distribution of any random variable is denoted by $P(.)$. Thus, the probability distribution of $Z$ is denoted by $P(Z)$ and the observed random sample is denoted as $Z_{obs}$. That is $Z_{obs} \sim P(Z)$. We denote the support of $Z$ as $\mathcal{Z} := supp(P(Z))$.   

\subsection{Model specification under network interference: Methodological examples}
\label{section: both applications}
We consider the potential-outcome framework, popularly known as the Rubin Causal Model (\cite{rubin1974estimating}). We consider an experimental setup in which the treatment assignment mechanism is given by $P(Z|X, G)$. In an experimental study, the dependence on the covariate and graph is known by design. Therefore, consider $P(Z|X, G)$ where the probability distribution is known and controlled by the experimenter. We take the population $\mathbb{P}$ fixed, and all randomness in the potential outcome function arises from design uncertainty. The definition and estimation of treatment effects in the presence of interference, such as the exposure contrast effect, hinge on structural or functional modeling assumptions. Without any assumption, it is impossible to identify, let alone estimate, spillover effects. There is a burgeoning literature on causal inference under network interference, exploring how the interference mechanism is specified in the outcome function. One example of such modeling assumptions used in this literature is of exposure mapping, which we present below. In Appendix \ref{section: example low beta order model}, we present another example of an outcome model for network interference, proposed in \cite{cortez2023exploiting}. These motivate the paper's focus on developing a model specification test to assess whether the proposed interference mechanism is correctly specified. 

\subsubsection{Exposure mapping specification}
\label{section: exposure mapping}
Exposure mappings are a low-dimensional functional representation of the treatment assignment vector. Given that plausibly potential outcomes are dependent on the entire treatment assignment vector, exposure mappings encode a functional representation of these dependencies. This facilitates not only the estimation of spillover effects but also their definition. \cite{savje2024causal} delineates the use of exposure mappings for both defining and encoding causal spillover mechanisms. The author suggests that this demarcation is helpful for defining exposure mappings and shows that, under some controlled forms of exposure-mapping misspecification in the definition, the estimators perform well. However, many practitioners are interested in uncovering the underlying causal interference mechanism, and a model specification testing procedure can help them assess whether the chosen exposure mapping is correctly specified for that mechanism. Here, we formally frame the misspecification of exposure mapping. Following \cite{savje2024causal}, we consider the misspecification error defined as $\epsilon_i(Z) = Y_i(Z) - E[Y_i(Z)\;|\;e_i = e_i(Z)]$ for $i\in [N]$. Thus, under no misspecification of the exposure mapping, we have $\epsilon_i = 0$ for any treatment assignment. We formulate this as a hypothesis testing statement to test for such misspecification. 
\begin{hypothesis}
\label{hypothesis: single exp miss}(Null hypothesis of no exposure mapping misspecification)
    \begin{equation}
        Y_i(Z) = E[Y_i(Z)\;|\;e_i(Z)]\quad \forall Z \in \{0,1\}^N, i \in [N].
    \end{equation}
\end{hypothesis}
Suppose a researcher posits several possible exposure mappings, each corresponding to a different causal interference mechanism. To infer the exposure mapping(s) that are correctly specified for the analysis, we can treat each exposure mapping's misspecification as a hypothesis test and construct a multiple-hypothesis testing statement. We formally state the multiple-hypothesis testing framework below. We assume that the set $E:= \{e_1,e_2,...e_M\}$ contains $M$ exposure mappings, where $M$ is fixed and finite.

\begin{hypothesis}
\label{hypothesis: multiple exposure}
(Multiple null hypotheses of no exposure mapping misspecification)
\begin{equation}
    \mathcal{H}_j: Y_i(Z) = E[Y_i(Z)\;|\;e_{ji}(Z)]\quad \forall Z \in \{0,1\}^N, i \in [N], j \in [M].
\end{equation}   
\end{hypothesis}
Because multiple exposure mappings are tested simultaneously, conducting separate specification tests for each mapping without adjustment would inflate the overall probability of making at least one false rejection. In this setting, false rejection of a correctly specified exposure mapping could lead the researcher to discard the true causal mechanism. To ensure that the overall probability of such an error remains controlled, we require procedures that bound the Family-wise error rate (FWER), that is, the probability of making one or more false rejections across all hypotheses (\cite{lehmann2005testing}). Controlling the FWER guarantees that, with probability at least $1-\alpha$, no correctly specified exposure mapping is incorrectly ruled out, thereby preserving the validity of subsequent causal conclusions. We consider an example below. Suppose a researcher postulates that an exposure mapping, for example, the proportion of treated neighboring units, captures the relevant exposure effect. However, the researcher remains uncertain about the depth of spillover effects, that is, whether interference arises only through immediate neighbors or also extends to second-degree and higher-order connections. \cite{zhang2025multiple} and \cite{zhong2024unconditional} consider the case of multiple hypothesis testing to test for spillover effect depth. We formulate this as a hypothesis testing framework below.
\begin{hypothesis}
\label{hypothesis: sequential test depth}
(Testing for depth of spillover effects)\\
    Consider $E := \{e_j: j \in [M]\}$. Suppose 
    \begin{equation}
        e_{ji} = (Z_i,f(Z_{{N_i}^1},{N_i}^1),..., f(Z_{{N_i}^j},{N_i}^j)) \quad \forall j \in [M], i \in [N].
    \end{equation}
     Here, ${N_i}^j = \{l: dist(i,l) = j\}$, that is, all the units at a distance of $j$ from the unit $i$. Consider the following set of hypotheses
     \begin{equation}
       \mathcal{H}_j:\: Y_i(Z) = E[Y_i(Z)\:|\:e_{ji}] \quad \forall j \in [M], i \in [N]. 
     \end{equation}     
\end{hypothesis}
The above hypothesis considers an increasing chain of conditioning subsets. This formulation is appropriate when researchers believe that any spillover effects, if present, decay monotonically with network distance, that is, higher-order neighbors exert progressively weaker influence (\cite{leung2022causal}). The above-stated hypothesis is structural and specifies how the outcome depends on exposure within a given neighborhood radius. Below, we provide examples of exposure mappings from the literature, presented as two-dimensional representations of the treatment vector.
\begin{example}
\label{eg: setting 1}
In \cite{cai2015social}, the authors are interested in studying information dissemination of a weather insurance product across villages in rural China. The authors hypothesize that the village's social network can create a word-of-mouth mechanism that drives up the adoption rate of the insurance product. In the article, the authors use exposure mapping $e_i(Z,G) =(Z_i,\frac{\sum_{j \in N_i} Z_j}{|N_i|})$. That is, the authors use the proportion of treated neighbors as a low-dimensional representation of the exposure effect.  
\end{example}
\begin{example}
\label{eg: setting 2}
In \cite{paluck2016changing}, the authors investigate the impact of peer-to-peer influence on anti-bullying behavior. In an interventional study, they classify students in the school into two categories: those connected with a social referent ($S$) and those without. Being connected with an influential student is postulated to have more influence. This leads us to the exposure mapping $e_i(Z,G) = (Z_i,\mathcal{I}(j \in N(i)\bigcap S: Z_j = 1))$. To contrast the influence of social referents, one can construct the following exposure mapping $e^c_i(Z, G) = (Z_i,\mathcal{I}(\nexists\; j \in N(i)\bigcap S:\; Z_j = 1))$.    
\end{example}
\begin{example}
\label{eg: setting 3}
In \cite{collazos2021hot}, the authors study the effect of ``hotspot policing", that is, intensive policing on streets identified as hotspots for crime. Since there may be spillover onto neighboring streets due to intensive policing on one street, the study needs to account for this effect. Given the setting's geographic nature, one can construct an exposure map based on radial distance from a treated street. Consider the following exposure mapping $e_i(Z,G) = (Z_i, I(\exists\; j: \;dist(i,j)\leq 500, Z_j =1))$.    
\end{example}

\subsection{Causal model specification as a hypothesis test}
\label{section: causal model selection}
We now define a general framework for performing model specification in causal inference under network interference. In particular, we focus on testing for misspecification of exposure mapping as described in Section \ref{section: exposure mapping}. We define the misspecification error as residual under the finite-sample setting, stated formally below. 
\begin{definition}
\label{definition: residual}
(Finite-sample residual) Consider an exposure mapping $e_i(Z, G)$. Then, we define the finite-sample residual as 
\begin{equation}
     R_i(Z) = Y_i(Z) - E_Z[Y_i(Z)\;|\;e_i(Z)].
\end{equation} 
where $E_Z[Y_i(Z)\;|\;e_i(Z)]$ represent the conditional mean response function. 
\end{definition}
We work with the design-based framework; hence, all stochastic variation in Definition \ref{definition: residual} arises from variation in treatment assignment $Z$. We can view this framework as a potential-outcome model conditional on other data-generating processes, such as sampling units from the observed population or observing covariates. Hence, we do not explicitly write $X$ in Definition \ref{definition: residual}. As stated before, we work within an experimental setup that allows the analyst to control the mechanism for treatment assignment. We formally write this below. 
\begin{assumption}
\label{assumption: experiment setup}
(Randomized experiment)
    Consider a given treatment assignment mechanism $P(Z)$ such that by design
    \begin{equation}
        Z \perp\!\!\!\perp Y(z) \quad \forall\; z\in \{0,1\}^N.
    \end{equation}
    Here, $P(Z)$ is known and is assumed to be positive for all $Z$. We consider $P(Z_i) \sim_{i.i.d.} Bernoulli(p)$.  
\end{assumption}
Common examples of randomized experiments include Bernoulli experiments and completely randomized experiments, where $Z \perp\!\!\!\perp X$. For ease of exposition, we restrict our attention to Bernoulli randomized trials. Under the null hypothesis of no model misspecification, we get $Y_i = E[Y_i|e_i]$ (see Hypothesis \ref{hypothesis: single exp miss}). This implies that under the null hypothesis of no model misspecification, we get $R_i(Z) = 0$ for any treatment assignment $Z$ and unit $i$ in the population, following Definition \ref{definition: residual}. We will later modify the residual definition such that any unit-level baseline effect, potentially heterogeneous, is separated out from the conditional mean response function. This baseline effect could depend on covariates and pre-treatment characteristics. The purpose of the residual is to capture the stochastic variation induced in the outcome model by the exposure mapping, which depends on the treatment assignment vector. Separating out the baseline effect yields the same null and does not equate to zero across all units. We explain this further using the example stated below. Consider the following linear model. 
\begin{example}
\label{example: GLM}
(Generalized Linear Model) Consider an exposure mapping of the form $e_i(Z) = (Z_i, f_i(Z_{-i}))$. 
\begin{equation}
  Y_i(Z) = \tau_dZ_i + \tau_{sp}f_i(Z_{-i}) + \beta_i\cdot X_i + c_i.  
\end{equation}    
\end{example}
Here, $\tau_d$ and $\tau_{sp}$ represent the direct and spillover effect, respectively. In the above model, $e_i(Z)$ is the correctly specified exposure mapping for the potential outcome model. We would like to propose a convention where the conditional mean response function is $\tau_dZ_i + \tau_{sp}f_i(Z_{-i})$. Thus, the modified residual $\tilde{R_i}(Z) = \beta_i\cdot X_i + c_i$ in this example. 

We first formally present our null hypothesis below in terms of $R_i(Z)$. Then, we define and present the null hypothesis in terms of the modified residual $\tilde{R}(Z)$. The key intuition is that the residual is invariant to the treatment assignment mechanism. This is because, if the exposure mapping model is correctly specified, all causal mechanisms should be captured in the conditional mean response function, and the residual should be independent of the treatment assignment, that is, $R_i(Z) \perp\!\!\!\perp Z \quad \forall\; Z \in \{0,1\}^N,\; i\in [N]$ under the null. This implies that the residuals remain constant and are invariant to changes in the treatment vector $Z$. We now formally define $\tilde{R}(Z)$ to include the baseline (heterogeneous) fixed effect as part of the residual.

\begin{definition}
\label{def: modified residual}
Consider $Y(Z)$. We define $\tilde{Y}(Z)$ such that
\begin{equation}
    Y_i(Z) = \tilde{Y}_i(Z) + c_i,
\end{equation} 
and $\tilde{Y}(\bar{0}) = 0$. Here, $c_i$ represent the baseline fixed constants.\\
We define $\tilde{R}(Z)$ as follows.
\begin{equation}
    \tilde{R}_i(Z) = Y_i(Z) - E[\tilde{Y}_i(Z)|e_i(Z)]. 
\end{equation}
\end{definition}
We interpret $\tilde{Y}(\bar{0})$ being set to $0$ as the causal outcome isolated to the exposure. We define the causal outcome attributed to the treatment to be $\tilde{Y}$, which is set to be zero if there is no treatment exposure. Since under the null, $R_i(Z) = 0$ for all the units $i$, we get $\tilde{R}_i(Z) = c_i$ for all the units $i$. Hence, $\tilde{R}(Z)$ is also invariant to $Z$. We present our main residual testing hypothesis below. 

\begin{mdframed}[leftmargin=0.1in,rightmargin=0.1in]
\begin{hypothesis}(Null of correctly specified interference mechanism model)
\label{hypothesis: main residual hypothesis}
\begin{equation}
    \mathcal{H}_0: \quad \tilde{R_i}(Z) = \tilde{R_i}(Z') \quad \forall\; Z,Z' \in \mathcal{Z}\;\; \forall\: i \in [N].  
\end{equation}   
\end{hypothesis}
\end{mdframed}
\begin{remark}
    Hypothesis \ref{hypothesis: main residual hypothesis} is our main hypothesis testing null framework that we will use to build the model specification procedure. The procedure is based on the residuals defined in the null hypothesis, and not on the potential outcome function directly, a key challenge that the paper addresses.   
\end{remark}

\section{Main result}
\label{section: main result}
We build towards a causal model specification procedure based on hypothesis testing. We present a hypothesis testing procedure to test if an exposure mapping is correctly specified, as defined in Hypothesis \ref{hypothesis: single exp miss}. For this, we build a randomization test in the spirit of Fisher's Randomization Test (FRT; \cite{fisher1966design}). One of the key features of FRT is the assumption-lean nature of the test, since it builds the randomization distribution of a test statistic using the design of the experiment $P(Z)$. Since $P(Z)$ is known to the analyst, the distribution is known exactly in finite samples. We provide a brief overview of the FRT procedure below.
\subsection{Fisher randomization test: Overview}
\label{sec: FRT}

The Fisher Randomization Test (FRT) is a hypothesis-testing procedure based on randomization inference. The key feature of FRT is testing for \textit{sharp} null-hypothesis violations. Under a sharp null of no treatment effect, $Y_i(1) = Y_i(0)$ for all units $i$ in the population. This is a stricter hypothesis than testing whether the Average Treatment Effect (ATE) is zero, that is, $E[Y_i(1)-Y_i(0)] = 0$. This is referred to as the \textit{weak} null in the literature. We present the FRT procedure below. 
\begin{proced}
\label{procedure FRT}  
Consider $(Z_{obs}, Y_{obs})$ as an experimental data.
\begin{enumerate}
    \item Consider a test statistic $T(Z,Y(Z))$.
    \item Define the following test statistic $T_z = T(z,Y_{obs})$, where $z \sim P(Z)$.
    \item Compute the test statistic $T_{Z_{obs}} = T(Z_{obs}, Y_{obs})$.
    \item Define and compute $pval = \frac{1}{L}\cdot \sum_{z_j \stackrel{\text{i.i.d.}}{\sim} P(Z):\: {j \in [L]}}I(T_{z_j} \geq T_{Z_{obs}})$.
\end{enumerate}
\end{proced}
Similarly, we can define and obtain p-value for two-sided test by taking $pval = 2\min(\frac{1}{L}\cdot \sum_{z_j \stackrel{\text{i.i.d.}}{\sim} P(Z):\: {j \in [L]}}I(T_{z_j} \geq T_{Z_{obs}}), \frac{1}{L}\cdot \sum_{z_j \stackrel{\text{i.i.d.}}{\sim} P(Z):\: {j \in [L]}}I(T_{z_j} \leq T_{Z_{obs}}))$. Here, no assumption is imposed on the test statistic and should be chosen according to its sensitivity to the sharp null hypothesis. Different test statistics differ in their power to test the null. A standard choice for testing no treatment effect is the Difference-in-means estimator, which is equal to $\sum_i(Y_i Z_i/N_t - Y_i(1-Z_i)/N_c)$, where $N_t$ is the number of treated units and $N_c$ is the number of control units. The choice of the number of re-draws $L$ is controlled by the analyst, and can in principle be chosen arbitrarily large. Hence, the p-values obtained from the procedure are exact. This provides a robust approach to randomization inference in experimental settings. Since it exploits the stochasticity of the design of the experiment, the inference holds validity in finite samples, and does not need an asymptotic regime (see \cite{imbens2015causal} for a detailed discussion).

\subsection{Causal model specification: A randomization  test}
We defined causal model specification under the potential outcome framework in Section \ref{section: causal model selection}. The main hypothesis that accounts for the testing of model misspecification of an exposure mapping is Hypothesis \ref{hypothesis: main residual hypothesis}. The hypothesis states that under the null of no model misspecification in the exposure mapping, the residuals are invariant to the treatment assignment vector. This is similar to the FRT, where, under the null of no treatment effect, the potential outcome model was invariant with respect to the treatment assignment vector. This imputation of missing potential outcomes helps build a randomization test based on the experimental design. In the same spirit, we will build a randomization hypothesis exploiting the imputation of residuals under the null. There has been a recent surge of literature on extending FRT to complex treatment effects. This includes testing for spillover effects when network interference is present (\cite{athey2018exact}; \cite{basse2019randomization}) and testing for heterogeneous treatment effects (\cite{ding2016randomization}). Notably, residuals have been used in randomization tests to construct test statistics that incorporate covariates in the potential outcome model (e.g., \cite{zhao2021covariate}). \\
Since our problem of interest is testing for exposure model misspecification, our work builds on randomization testing but is conceptually distinct from the literature on FRT extensions for testing treatment effects. The null hypothesis concerns itself with the true residuals, not the potential outcomes. Since we do not directly observe the residuals, we do not know the true value of $\tilde{R}_{obs}$ (defined to be $\tilde{R}(Z_{obs})$). In this section, we present a proof-of-concept for the oracle procedure that assumes that the true value of the observed residual, that is $\tilde{R}_{obs}$, is known. Consider the following procedure.
\begin{proced}
\label{procedure oracle}    
Consider $(Z_{obs}, Y_{obs})$ as the observed data from an experiment.
\begin{enumerate}
    \item Consider a test statistic $T(Z,\tilde{R}(Z))$.
    \item Define the test statistic $T_z = T(Z=z, \tilde{R}_{obs})$ where $z \sim P(Z)$.
    \item Compute $T_{Z_{obs}}$ which is equal to $T(Z_{obs}, \tilde{R}_{obs})$.
    \item Define and compute $pval(Z_{obs}, \tilde{R}_{obs}) = \frac{1}{L}\cdot \sum_{z_j \stackrel{\text{i.i.d.}}{\sim} P(Z):\: {j \in [L]}}I(T_{z_j} \geq T_{Z_{obs}})$.
\end{enumerate}
\end{proced}
Here, $L$ denotes the number of redraws from the distribution of $Z$. This is controlled by the analyst, and can be taken as large as deemed necessary for arbitrary precision. Thus, the p-value obtained from the above procedure is exact, similar to FRT. Alternatively, we can use p-value as $\frac{1+\sum_{z_j \stackrel{\text{i.i.d.}}{\sim} P(Z):\: {j \in [L]}}I(T_{z_j} \geq T_{Z_{obs}})}{1+L}$ for an exact estimate. Below, we formalize and prove that the p-values obtained from Procedure \ref{procedure oracle} hold validity in finite samples. 
\begin{theorem}
\label{theorem: oracle method validity}
    Let $(Z_{obs}, Y_{obs})$ be the observed data from a randomized experiment such that $Z_{obs} \sim P(Z)$ and $\tilde{R}_{obs} = \tilde{R}(Z_{obs})$. Consider an exposure mapping $e(Z)$ and a test statistic $T(Z, \tilde{R}(Z))$. Consider the null hypothesis that the exposure mapping $e(Z)$ is correctly specified in Hypothesis \ref{hypothesis: main residual hypothesis}, denoted by $\mathcal{H}_{0}$. Assuming that $\tilde{R}_{Z_{obs}}$ is known, Procedure \ref{procedure oracle} holds validity at a given level $\alpha \in (0,1)$. That is, 
    \begin{equation}
        E[I(pval(Z_{obs}, \tilde{R}_{obs}) \leq \alpha)\: |\: \mathcal{H}_0] \leq \alpha. 
    \end{equation}
    Here, the expectation is with respect to $P(Z_{obs})$. 
\end{theorem}
\begin{proof}
    Refer to Appendix \ref{proof: oracle method validity}.
\end{proof}
The proof of Theorem \ref{theorem: oracle method validity} is a classic result from randomization inference. For completeness, we provide its proof in Appendix \ref{proof: oracle method validity}. Following \cite{lehmann2005testing}, one can also make the validity statement in equality under expectation, as opposed to inequality. This makes randomization inference exact, making the use of experimental design for randomization inference attractive. 

\subsection{Test statistic}
\label{section: test statistic}
We do not impose any assumption on the test statistic for Procedure \ref{procedure oracle} to be valid. But the test's power depends on the choice of test statistic. A test statistic should be sensitive to deviations from the null for Procedure \ref{procedure oracle} to have power. Here, we propose a test statistic that captures deviations from the Hypothesis \ref{hypothesis: main residual hypothesis}. Under the null hypothesis, the residuals should be free of treatment effects, and the exposure mapping should capture all dependence on the treatment assignment of other units. We take this opportunity to use the network $G$ amongst the units in the population to build the test statistic. We consider the correlation between the treatment assignment vector and the residuals weighted by the graph. Any systematic correlation between residuals and the treatment status of neighboring nodes would therefore indicate that the exposure mapping fails to capture the relevant channels through which interference operates. The network $G$ naturally encodes these potential channels of dependence, serving as the adjacency structure for detecting residual-treatment correlation. We will denote this measure by $T_{GC}$, abbreviated for `Graphical Correlation'. We formally present this metric below:   
\begin{equation}
T_{GC} = \frac{\sum_{i} \sum_{j} G_{ij}Z_i (\tilde{R}_j - \bar{\tilde{R}})}
         {||(\tilde{R} - \bar{\tilde{R}})||\;||Z||}.
\end{equation}
By weighing edges according to $G_{ij}$, we focus specifically on associations between a unit's treatment and the residuals of its neighbors, which is the type of dependence that should vanish under the null but appear under violations due to unaccounted spillovers or mis-specified exposure mappings.\\
It should be noted that $T_{GC}$ can only detect network dependence of the residuals up to an order of one. That is, if there are network spillover effects at a distance of two or higher, $T_{GC}$ will not hold any strength in detecting a spillover mechanism. For example, consider $e_{ki} = (Z_i, f(Z_{{N_i}^k}, {{N_i}^k}))$ for $1<k \leq N$. These exposure mappings explicitly capture the spillover mechanism at a distance $k$ in the graph $G$. We define a generalized version of the residual graphical correlation test statistic below. Consider the matrix that contains entries of all units that are exactly at a distance of $k$ in the graph, denoted by $G^{(k)}$. That is, $G^{(k)} = [I(dist(i,j) = k)]_{N\times N}$ where $i$ and $j$ are two units in the population. Consider $1\leq k \leq N$. Then, 
\begin{equation}
\label{eq: test stat def general}
{T_{GC}}^{(k)} = \frac{\sum_{i} \sum_{j} G_{ij}^{(k)}Z_i (\tilde{R}_j - \bar{\tilde{R}})}
         {||(\tilde{R} - \bar{\tilde{R}})||\;||Z||}.
\end{equation}
When $k = 1$, we recover the original statistic $T_{GC}$. That is, ${T_{GC}}^{(1)}$ is equal to $T_{GC}$. By varying $k$, we obtain a family of test statistics, each sensitive to dependence at a particular graph distance.

\section{Estimated residual randomization test: Asymptotic validity}
\label{section: est test}
 
In the above, we present the case for the validity of the oracle procedure, assuming residuals are known, as a proof-of-concept. Since we do not observe these residuals for the data, a direct application of FRT is not possible. We build a procedure by considering a test statistic based on the estimated conditional mean response function. We present the Procedure for estimating the p-value.
\begin{proced}
\label{procedure estimated}    
 Let $(Z_{obs}, Y_{obs})$ be the observed data from an experiment, and $\phi(.)$ be a given function in $\mathbb{R}^{d}$. Define $\hat{\bar{\beta}}$ as the OLS coefficient of $Y(Z) \sim [1,\phi(e(Z))]$, and $\hat{\beta} := \hat{\bar{\beta}}_{[2,...,(d+1)]}$.  
 \begin{enumerate}
     \item Consider $Y_{obs} \stackrel{\tiny \text{OLS}}{\sim} [1,\phi(e_{obs})]$ and obtain estimate of $\tilde{R}_{obs}$ as $Y_{obs} - \hat{\beta}_{obs}\cdot \phi(e_{obs})$.
     \item Take $\hat{Y}_i(Z) = \tilde{R}_{obs_i} + \hat{\beta}_{obs}\cdot \phi(e(Z))$ where $\hat{\tilde{R}}(Z)$ is the residual estimate obtained from the OLS model $\hat{Y}(Z) \sim [1,\phi(e(Z))]$.
     \item Consider the test statistic $T_z = T(z,\hat{\tilde{R}}(z))$, where $z \sim P(Z)$.
     \item Compute $T_{Z_{obs}}$ which is equal to $T(Z_{obs}, \hat{\tilde{R}}(Z_{obs}))$.
     \item Define and compute $\hat{pval}(Z_{obs}, \tilde{R}_{obs}) = \frac{1}{L}\cdot \sum_{z_j \stackrel{\text{i.i.d.}}{\sim} P(Z):\: {j \in [L]}}I(T_{z_j} \geq T_{Z_{obs}})$.
 \end{enumerate}
\end{proced}
Consider a science $\{Y_i(z)\}$ for all $z \in \{0,1\}^N$. We define the following population loss function.
\begin{definition}
\label{definition: MSE}
(Finite-population mean-square error) Consider a model space on $e$ as $\mathcal{M}$. Let $I \sim Unif\{1,...,N\}$. Then, given $m_i \in \mathcal{M}^N$, we define
    \begin{equation}
    \begin{aligned}
        L_{\mathbb{P}_N,Z}(m) &= E_{I,Z}[(Y_I(z) - m_I(e_I))^2], \\
        L_{\mathbb{P}_N,Z} &= min_{m\in \mathcal{M}^N}  \;L_{\mathbb{P}_N,Z}(m).
    \end{aligned}
    \end{equation}
\end{definition}

Here, $\mathcal{M}$ is a class of functions with a given $e$ as domain. Consider the following 
\begin{lemma} 
\label{lemma: MSE minimizer}
    $L_{\mathbb{P}_N,Z}$ has a minimizer such that 
    \begin{equation}
        argmin_{m \in \mathcal{M}^N} E_{I,Z}[(Y_I(z) - m_I(e_I))^2] = argmin_{m \in \mathcal{M}^N} E_{I,Z}[(E_Z[Y_I|e_I] - m_I(e_I))^2]. 
    \end{equation}
    and if $E_Z[Y_i|e_i] \in \mathcal{M}_i$ for all $i \in [N]$, then $m_i(e_i) = E_Z[Y_i|e_i]$ a.s.
\end{lemma}
\begin{proof}
    Refer to Appendix \ref{proof: MSE minimizer}.
\end{proof}

Therefore, the MSE minimizer provides the closest approximation to the conditional mean response function within the chosen model class. If $E_Z[Y_i|e_i] \in \mathcal{M}_i$, then $L_{\mathbb{P}_N,Z} = \frac{1}{N}\sum_i\sum_{z}(Y_i(z)- E[Y_i|e_i])^2\cdot P(Z=z)$. Consider the finite population loss function (defined in Definition \ref{definition: MSE}).  We present the sample estimator below for the finite-population MSE loss estimand. We show that the sample estimator is unbiased. We also show that the sample estimator, given a model function, converges in probability to the population estimand corresponding to that function. We present this formally below. In this article, we restrict our attention to the ordinary least squares setting and assume the conditional mean response function to be approximable by a parametric linear basis function model, as presented in Assumption \ref{ass: OLS} below. We also make two assumptions about the potential outcome function and the network-generating process, similar to Conditions $3$ and $5$ in \cite{aronow2017estimating}. 
\begin{assumption}
\label{assumption: bdd outcome}
(Bounded outcomes and exposures) There exists $B \in \mathbb{R}$, and there exists $\epsilon_z \in \mathbb{R}$ for all $z$ such that 
\begin{equation}
\begin{aligned}
    |Y_i(z)| &\leq B \quad \forall i \in [N], \forall z\in \{0,1\}^N,\\
    ||e_i(z')|| &\leq b_z\quad\forall||z-z'||\leq \epsilon_z.
\end{aligned}
\end{equation}
That is, the potential outcome functions are bounded, and exposure mappings are locally bounded, both uniformly in N.     
\end{assumption}
As an illustration for the bounded exposure mappings, consider the settings in Example \ref{eg: setting 1}, \ref{eg: setting 2}, and \ref{eg: setting 3}. The exposure mappings considered are globally bounded between 0 and 1. 
\begin{assumption}
\label{assumption: bdd degree}
(Bounded degree of graph)
    Consider $\Delta(G)$ as the maximum degree of the graph. Then, 
    $\Delta(G) \leq \kappa$ for some $\kappa \in \mathbb{N}$. 
\end{assumption}
This assumption imposes limited dependence on the potential outcomes through exposure mapping. 
\begin{remark}
    The randomness of the process arises from design uncertainty in the treatment assignment mechanism. This contrasts with the standard sampling uncertainty under the super-population perspective. Our analysis fixes the population size $N$ and considers an experiment on the population $\mathbb{P}_N$. As we increase the sample size $N$, we show that the corresponding estimators of the MSE loss function converge. 
\end{remark}
Consider a given parametric linear basis model on the domain $e$, as defined below.
\begin{definition}
\label{def: model space basis}
    (Linear basis model) Consider a continuous function $\phi:\mathbb{R}^k\rightarrow\mathbb{R}^d$. Define 
    \begin{equation}
        \mathcal{M}_e = \{ \bar{\beta}\cdot [1, \phi(e_i)] : \bar{\beta} = [\beta_c,\beta],\:\beta_c \in \mathbb{R}, \beta \in \mathbb{R}^{d}\}.
    \end{equation} 
\end{definition}
That is, a given value of $\bar{\beta} \in \mathbb{R}^{d+1}$ corresponds to the function $m(e) = \bar{\beta} \cdot [1, \phi(e)]$. The following Lemma establishes the asymptotic consistency of the sample estimator for the MSE loss function (as defined in Definition \ref{definition: MSE}) under the parametric linear basis model, under the null. 
\begin{lemma}
\label{lemma: consistency}
    Consider the null hypothesis of correctly specified exposure mapping $e(Z)$, as defined in Hypothesis \ref{hypothesis: main residual hypothesis}, defined as $\mathcal{H}_0$. Given $m \in \mathcal{M}_e$, we define the following estimator 
    \begin{equation}
        \hat{L}_{e_{obs}}(m) = \frac{1}{N}\cdot \sum_i(Y_i({e_{obs}}_i) - m({e_{obs}}_i))^2.
    \end{equation}
    Then, under Assumption \ref{assumption: experiment setup}, \ref{assumption: bdd outcome}, \ref{assumption: bdd degree} and $\mathcal{H}_0$, we get
    \begin{equation}
        \hat{L}_{e_{obs}}(m) - L_{\mathbb{P}_N,Z}(m) \xrightarrow{p} 0. 
    \end{equation}
\end{lemma}
\begin{proof}
    Refer to Appendix \ref{proof: consistency}. 
\end{proof}
We can parameterize the model space with the basis coefficients to obtain $\hat{L}_{e_{obs}}(m) = \hat{L}_{e_{obs}}(\bar{\beta})$ given $\bar{\beta} \in \mathbb{R}^{d+1}$. Thus, Lemma \ref{lemma: consistency} concludes $\hat{L}_{e_{obs}}(\bar{\beta})$ converges in probability to $L_{\mathbb{P}_N,Z}(\bar{\beta})$. This convergence holds under the null hypothesis $\mathcal{H}_0$ as defined in Hypothesis \ref{hypothesis: main residual hypothesis}. We now prove that the corresponding minimizer-argument estimator also converges to the true population minimizer. We note that $argmin_{\bar{\beta}}\: \hat{L}_{e_{obs}}(\bar{\beta})$, by definition, matches the ordinary least squares estimator obtained from linear regression for the basis function.
\begin{assumption}
\label{ass: design matrix pd}
$E_{I,Z}[[1,\phi(e)][1,\phi(e)]^T]$ is positive definite uniformly in $N$. 
\end{assumption}
That is, we assume the design matrix is uniformly invertible as the sample size increases, and hence, admits a positively lower-bounded eigenvalue support. 
\begin{lemma}
\label{lemma: argmin consistency}
We define $\hat{\bar{\beta}}_{obs} = argmin_{\bar{\beta}}\: \hat{L}_{e_{obs}}(\bar{\beta})$, and $\bar{\beta}_0 = argmin_{\bar{\beta}}\: L_{\mathbb{P}_N,Z}(\bar{\beta})$. Then, under Assumption \ref{assumption: experiment setup}, \ref{assumption: bdd outcome}, \ref{assumption: bdd degree}, \ref{ass: design matrix pd}, and $\mathcal{H}_0$
\begin{equation}
    \hat{\bar{\beta}}_{obs} - \bar{\beta}_0 \xrightarrow{p} 0. 
\end{equation}
\end{lemma}
\begin{proof}
    Refer to Appendix \ref{proof: argmin consistency}.
\end{proof}
Here, $\bar{\beta}_0 = [\beta_{0_c},\beta_0]$ for $\beta_{0_c} \in \mathbb{R}$, and $\beta_0 \in \mathbb{R}^d$. Similarly, $\hat{\bar{\beta}}_{obs} = [\hat{\beta}_{obs_c}, \hat{\beta}_{obs}]$ with $\hat{\beta}_{obs_c} \in \mathbb{R}$, and $\hat{\beta}_{obs} \in \mathbb{R}^d$.  
\begin{remark}
    Consider $Z,Z_{obs} \sim_{i.i.d.} P(Z)$, where $(Z_{obs},Y(Z_{obs}))$ is the observed data. Here, $(Z, Y(Z))$ is counterfactual data that is not observed. Then, Lemma \ref{lemma: argmin consistency} shows that both the OLS coefficients $\hat{\bar{\beta}}_{obs}$ and $\hat{\bar{\beta}}(Z, Y(Z))$, obtained from regressing $Y(Z)$ on $\phi(e(Z))$, converge to the same constant ($\bar{\beta}_0$, here) in probability under the null of residual invariance. Here, $\phi(.)$ is any continuous basis function model. Hence, even though direct imputation of residuals is not feasible, asymptotic imputation of OLS coefficients of counterfactuals is feasible using the observed data as $\hat{\bar{\beta}}_{obs}$.    
\end{remark}
As noted above, asymptotic imputation of the OLS coefficients of the counterfactuals is feasible. To recover the counterfactual residual, we assume that the conditional mean response function for each unit lies within a bounded ball around the parametric linear basis model space, as stated below. 
\begin{assumption}
\label{ass: OLS}
(Bounded discrepancy variation)
 There exists $s \geq 0$ such that,   
    \begin{equation}
        \sup_{z\in\{0,1\}^N}1/N\sum_i||E[\tilde{Y}_i|e_i]  - \beta_0\cdot\phi(e_i)||^2 \leq s^2.
    \end{equation}
\end{assumption}
Here, $\beta_0$ is the population OLS minimizer excluding the baseline constant, as defined in Lemma \ref{lemma: argmin consistency}. We now prove the asymptotic validity of Procedure \ref{procedure estimated}. 
\begin{theorem}
\label{theorem: estd asymptotic validity}
    Consider $(Z_{obs}, Y_{obs})$ to be the observed data from a randomized Bernoulli experiment where $Z_{obs} \sim P(Z)$, and $Y(Z_{obs}) = Y_{obs}$. Let $e(Z)$ be an exposure mapping, and a test statistic $T(Z, \hat{\tilde{r}})$, where $\hat{\tilde{r}}$ is the estimated modified residual obtained from the OLS model $Y(Z) \sim [1,\phi(e(Z))]$. We assume that the functions $T_x(y): = T(x,y)$ are uniformly Lipschitz in $y$, with Lipschitz constant $C/\sqrt{N}$ under the null $\mathcal{H}_0$. Consider a null hypothesis that the exposure mapping $e(Z)$ is correctly specified (defined in Hypothesis \ref{hypothesis: main residual hypothesis} as $\mathcal{H}_0$). Then, under Assumption \ref{assumption: experiment setup}, \ref{assumption: bdd outcome}, \ref{assumption: bdd degree}, \ref{ass: design matrix pd}, and \ref{ass: OLS}, Procedure \ref{procedure estimated} holds validity asymptotically with a corrective factor. That is, given $\alpha$
    \begin{equation}
        \lim_{\epsilon\rightarrow 0}\limsup_{N \to \infty} E_{Z_{obs}}[I(\hat{pval}_{\epsilon}(Z_{obs}, \hat{\tilde{R}}_{obs}) \leq \alpha)\: |\: \mathcal{H}_0] \leq \alpha. 
    \end{equation}
    Here, the expectation is with respect to $P(Z_{obs})$, and for any $\epsilon > 0$,
    \begin{equation}
        \hat{pval}_{\epsilon} = P_Z(T(Z,\hat{\tilde{r}}_{obs})\geq T(Z_{obs},\hat{\tilde{r}}_{obs})-C's-\epsilon).
    \end{equation}
    for some $C'>0$. 
\end{theorem}
\begin{proof}
Refer to Appendix \ref{proof: estd asymptotic validity}.
\end{proof}
In the above limit, for any small inflation of Type I error, there exists an $\epsilon$ small that controls it at that inflation level. Since inflation is arbitrary, $\epsilon$ can be taken arbitrarily small in practice. Alternatively, one can add tie-breaking mechanisms to deal with probability mass concentration at $\epsilon = 0$ (refer to \cite{lehmann2005testing} or \cite{ritzwoller2024randomization} for further details). The above considers the setting of a single hypothesis test. In Section \ref{section: both applications}, we also see examples of multiple model specification testing procedures. We consider the case of nested hypotheses in Appendix \ref{sec: nested testing}. \cite{zhang2025multiple} demonstrates the feasibility of a general methodology for multiple randomization tests in causal inference.

\subsection{Test statistic and sensitivity analysis}
\label{sec: proc 2 overview}
We now show the estimated graphical correlation test statistic, $T_{GC}^{(k)}$ as defined in Equation (\ref{eq: test stat def general}), in Procedure \ref{procedure estimated}, and then prove that it is uniformly Lipschitz. This verifies the assumption imposed in Theorem \ref{theorem: estd asymptotic validity}. For the validity of a Fisherian-style randomization test, the choice of test statistic is arbitrary. The more sensitive the test statistic is to deviations from the null hypothesis, the greater the method's power will be. Hence, in Theorem \ref{theorem: oracle method validity}, where we prove the validity of the oracle method (Procedure \ref{procedure oracle}), the choice of test statistic is arbitrary. That is, if the (modified) residuals, as defined in Definition \ref{def: modified residual}, were known for the observed data, then Procedure \ref{procedure oracle} controls Type I error rates at a chosen level of significance. Since the residuals are unknown, we substitute for estimated residuals. This opens up two lines of inquiry: studying the asymptotic validity of the randomization testing procedure and studying deviations from the null due to model misspecification of the conditional mean response function. We restrict the test statistic to be a Lipschitz function and show guarantees of a Fisher randomization test. We now present the test statistic proposed for Procedure \ref{procedure estimated} and show that it is Lipschitz with Lipschitz constant $C/\sqrt{N}$. 

\begin{proposition}
\label{prop: test stat lip}
    Consider the test statistic
    \begin{equation}
     \label{eq: test stat moran beta}
     {T_{\beta}}^{(k)}(Z, \tilde{r}) = 
    \frac{\sum_{i} \sum_{j} {G_{ij}}^{(k)}Z_i\tilde{r}_j}{\sqrt{\sum_{i} (\tilde{r}_i)^2}\sqrt{\sum_i Z_i}}.
    \end{equation}
    where $\tilde{r}_i = {Y_i(Z)} - \beta\cdot{\phi(e_i(Z))} - 1/N\sum_i ({Y_i(Z)} - \beta\cdot{\phi(e_i(Z))})\cdot 1_N$. Suppose $\exists\: c_0> 4s^2$ such that $\operatorname{Var}(\tilde{r}_{obs}) > {c_0}$. 
    Then, $T_{\beta}^{(k)}(Z,\tilde{r})$ is Lipschitz in $\tilde{r}$ under the null $\mathcal{H}_0$. That is, $|{T_{\beta}}^{(k)}(Z, \tilde{r}) - {T_{\beta}}^{(k)}(Z, \tilde{r}')| \leq C\frac{||\tilde{r}-\tilde{r}'||}{\sqrt{N}}$ for some $C >0$.
\end{proposition}
\begin{proof}
    Refer to Appendix \ref{proof: lipschitz}.
\end{proof}
Put $k=1$ to obtain the first-order exposure mapping specification test statistic.  
\begin{remark}
    For the test statistic to be well-defined and numerically stable, we assume the empirical variance of the observed linear projection estimate, $\operatorname{Var}(\tilde{r}_{obs})$, to be bounded away from zero. As a practical alternative, the denominator may be truncated below by replacing it with $max(c_0, \operatorname{Var}(\tilde{r}_{obs}))$, where $c_0 > 0$ is of analyst choice (e.g., $c_0 = 0.01$).   
\end{remark}
We treat the discrepancy metric bound $s$ stated in Assumption \ref{ass: OLS} as a sensitivity parameter. As seen in Theorem \ref{theorem: estd asymptotic validity}, we get the $\epsilon$ perturbed p-value as $P_Z(T(Z,\hat{r}_{obs})-T(Z_{obs},\hat{r}_{obs}) \geq - C's-\epsilon)$ (refer Appendix \ref{proof: estd asymptotic validity}). Since $\{ P_Z(T(Z,\hat{r}_{obs})-T(Z_{obs},\hat{r}_{obs}) \geq -\epsilon) > \alpha\}\implies \{P_Z(T(Z,\hat{r}_{obs})-T(Z_{obs},\hat{r}_{obs}) \geq - C's-\epsilon) > \alpha \}$, failure to reject the null stands robust to potential model misspecification in the conditional mean response function when $s>0$. When rejecting the null, the rejection may be due to underlying model misspecification without accounting for the corrective term, which is scaled by the sensitivity parameter $s$, that is, $C's$. Not accounting for this would lead to Type I error inflation, and hence, reporting a rejection of the null should be accompanied by the sensitivity parameter. This can be obtained with linear grid search by observing when the decision flips. Reporting the parameter indicates how much model discrepancy is allowed for the results to remain robust.     


\section{Simulation study}
\label{section: sim study}
In this section, we present a Monte Carlo simulation study of our proposed testing procedure. We present our data-generating process and the general simulation setup. We consider a potential outcome model that is linearly separable in the direct and indirect effects of first- and higher-order variables. In our analysis, we do not incorporate covariates and define the potential outcome using exposure mapping. We consider different exposure-mapping specifications relative to the true exposure mapping from the data-generating process to assess the test's power. Consider the following two potential outcome models, where spillover effects are of first order.
\begin{equation}
\label{eq: outcome_prop}
I: Y_{i}(\mathbf{z}) =
\begin{cases} 
    {\tau_{\text{direct}}} \cdot z_{i} + {\tau_{\text{spill}_1}} \cdot \frac{\sum_{j=1}^{N}G_{ij} \cdot z_{j}}{\sum_{j=1}^{N}G_{ij}} + \epsilon_{i} \qquad \sum_{j=1}^{N}G_{ij} > 0, \\
    {\tau_{\text{direct}}} \cdot z_{i} + \epsilon_{i} \qquad \qquad \qquad \qquad \qquad \quad \text{otherwise}.
\end{cases}
\end{equation}
and
\begin{equation}
\label{eq: outcome_prop_sq}
II: Y_{i}(\mathbf{z}) =
\begin{cases} 
    {\tau_{\text{direct}}} \cdot z_{i} + {\tau_{\text{spill}_1}} \cdot \bigg(\frac{\sum_{j=1}^{N}G_{ij} \cdot z_{j}}{\sum_{j=1}^{N}G_{ij}}\bigg)^2 + \epsilon_{i} \qquad \sum_{j=1}^{N}G_{ij} > 0, \\
    {\tau_{\text{direct}}} \cdot z_{i} + \epsilon_{i} \qquad \qquad \qquad \qquad \qquad \quad \text{otherwise}.
\end{cases}
\end{equation}
Here, $\epsilon_i$ is generated from a standard normal variable and represents a heterogeneous baseline effect when there is no treatment. It is independent of any covariate characteristics. We also consider a potential outcome model in which spillover effects of both first and second orders are present, as shown below. 
\begin{equation}
\label{eq: outcome_prop_sec}
\begin{aligned}
III: Y_{i}(\mathbf{z}) &=
    {\tau_{\text{direct}}} \cdot z_{i} + {{\tau_{\text{spill}}}_1} \cdot \frac{\sum_{j=1}^{N}G_{ij} \cdot z_{j}}{\sum_{j=1}^{N}G_{ij}} I(\sum_{j=1}^N G_{ij} > 0) +\\
    &\hspace{4cm}{{\tau_{\text{spill}}}_2} \cdot \frac{\sum_{j=1}^{N}G_{ij}^{(2)} \cdot z_{j}}{\sum_{j=1}^{N}G_{ij}^{(2)}}I(\sum_{j=1}^N G^{(2)}_{ij} > 0) + \epsilon_{i}.
\end{aligned}
\end{equation}
Consider the following specifications of the simulation setup.
\begin{enumerate}
    \item Experiment: We consider a Bernoulli design in our setup. That is, \\$P(Z = z) = \Pi_i P(Z_i = z_i)$. Assume that $P(Z_i = z_i) = p$. We get that $P(Z=z) = p^{\sum z}(1-p)^{(1-\sum z)}$. In our study, we assume the probability of a unit being treated is $p = 0.5$. 
    \item Population: We fix the population size to $500$ and $1000$. 
    \item Network: We assume the network-generating process follows a small-world network. A small-world network is characterized by two parameters: the number of initial connections, $K$, and the probability of rewiring, $p_{rw}$. This generates a $ K$-regular graph and then rewires its edges with probability $p_{rw}$ to enhance clustering. We take $K = 3$ and $p_{rw} = 0.1$. \\
    We also use a stochastic block model with $5$ blocks of size $100$ (for a population size of $500$) and $5$ blocks of size $200$ (for a population size of $1000$). The preference matrix, which specifies the probability of an edge between two groups, has diagonal entries $\{0.02, 0.01, 0.01, 0.01, 0.03\}$ and all off-diagonal entries as $0.005$.    
    \item Treatment effect: We take $\tau_{direct} = \{0, 2\}$, $\tau_{spill_1}$ as $1.5$ and keep $\tau_{spill_2}$ as $1.2$. 
    \item Exposure mapping: We consider a true exposure mapping of the data-generating process in Equation (\ref{eq: outcome_prop}) and (\ref{eq: outcome_prop_sq}). For $i \in [N]$,
     \begin{equation}
        e_{1i} = \big(Z_i,\frac{1}{|N(i)|} \sum_{j \in N(i)} Z_j\big).
    \end{equation} 
    We first consider the potential mis-specification of not accounting for any spillover effect to test for. For this, we take the exposure mapping below. For $i \in [N]$,
    \begin{equation}
        e_{2i} = (Z_i).    
    \end{equation}
     \item Test statistic: We consider the test statistic ${T_{\beta}}^{(1)}$ and ${T_{\beta}}^{(2)}$ (see Section \ref{sec: proc 2 overview}). 
\end{enumerate}

We generated the null distribution of the test statistic by sampling 500 treatment assignment vectors from a Bernoulli experiment with a probability of being treated equal to 0.5. Our chosen level of significance is 0.05. We calculate the p-value for a two-tailed test. We replicated the above setup 1500 times. We first present the results for the oracle procedure, in which the true value of the observed residual is known. We test for the (correctly specified) exposure mapping $e_1$ for both data-generating processes (DGP) \hyperref[eq: outcome_prop]{I} and \hyperref[eq: outcome_prop_sq]{II}. This gives us Type I error rates. Subsequently, we test for misspecified exposure mapping $e_2$ for DGP~\hyperref[eq: outcome_prop]{I} and \hyperref[eq: outcome_prop_sq]{II}. $e_2$ is misspecified because it fails to account for spillovers. We use this as a first check to demonstrate the proposed test's power. We also consider DGP~\hyperref[eq: outcome_prop_sec]{III}, in which both first- and second-order spillover effects are present. We test whether $e_1$ is correctly specified, i.e., that it accounts only for the first-order effect. Hence, rejection rates again demonstrate the test's power here. We then repeat the same setup for the case in which the observed residuals are unknown. We estimate the residuals using standard linear regression, treating the exposure mappings as regressors, taking the basis maps as identity functions. We impute the counterfactual potential outcome function and re-fit the model to obtain the counterfactual residuals.  

\begin{table}[H]
\caption{The Table displays values of rejection rates for the oracle method under Hypothesis \ref{hypothesis: main residual hypothesis}. The column (I), (II), and (III) represents the data-generating mechanism in Equation (\ref{eq: outcome_prop}), (\ref{eq: outcome_prop_sq}), and (\ref{eq: outcome_prop_sec}), respectively. If not specified, we use the test statistic ${T_{\beta}}^{(1)}$ for the columns. The columns marked with `$\dagger$' shows Type I error.\\[0.05cm]}
\label{table: oracle method}
\centering
{ 
\begin{tabular}{cccccccccc}
\hline
 \multirow{3}{*}{\small \shortstack{ Network \\Distribution}} &  \multirow{3}{*}{\small \shortstack{Direct \\Effect}} &  \multirow{3}{*}{\small \shortstack{Population \\Size}} & \multicolumn{6}{c}{Exposure Mapping} \\
 \cline{4-10}
  &  &  & \multicolumn{3}{c}{(I)} &\multicolumn{3}{c}{(II)} & (III)  \\
  \cmidrule(lr){4-6}\cmidrule(lr){7-9}\cmidrule(lr){10-10}
  &  &  &  ${e_1}^{\dagger}$ & ${e_1(\small{T_{\beta}^{(2)}}})^{\dagger}$ & $e_2$ & ${e_1}^{\dagger}$ & ${e_1(\small{T_{\beta}^{(2)}}})^{\dagger}$& $e_2$ & ${e_1(\small{T_{\beta}^{(2)}}})$  \\[0.1cm]
  \hline
 \multirow{4}{*}{\small \shortstack{Small world \\network}}& 2 & 500 & 0.055 & 0.045 & 1.000 & 0.047 & 0.063 & 1.000 & 0.911\\[0.1cm]
 & 2 & 1000 & 0.046 & 0.052 & 1.000 & 0.048 & 0.057 & 1.000 & 0.997\\[0.1cm]
 & 0 & 500 & 0.043 & 0.055 & 1.000 & 0.034 & 0.047 & 1.000 & 0.912\\[0.1cm]
 & 0 & 1000 & 0.040 & 0.055 & 1.000 & 0.043 & 0.057 & 1.000 & 0.995\\[0.1cm]
\hline
\multirow{4}{*}{\small \shortstack{Stochastic block \\model}}& 2 & 500 & 0.037 & 0.051 & 1.000 & 0.051 & 0.044 & 1.000 & 0.956\\[0.1cm]
 & 2 & 1000 & 0.056 & 0.049 & 1.000 & 0.044 & 0.047 & 1.000 & 0.771\\[0.1cm]
 & 0 & 500 & 0.046 & 0.056 & 1.000 & 0.050 & 0.043 & 1.000 & 0.947\\[0.1cm]
 & 0 & 1000 & 0.053 & 0.055 & 1.000 & 0.045 & 0.045 & 1.000 & 0.784\\[0.1cm]
 \hline
\end{tabular}
}
\end{table}


\begin{table}[H]
\caption{The Table displays values of rejection rates for the estimated residual method under Hypothesis \ref{hypothesis: main residual hypothesis}. As before, the column (I), (II), and (III) represents the data-generating mechanism in Equations (\ref{eq: outcome_prop}), (\ref{eq: outcome_prop_sq}), and (\ref{eq: outcome_prop_sec}), respectively and we take the test statistic to be $T_{\beta}^{(1)}$ unless mentioned otherwise. As before, columns marked with `$\dagger$' show the Type I error rate of the method. 
\\[0.05cm]}
\label{table: est residual method}
\centering
{ 
\begin{tabular}{cccccccccc}
\hline
 \multirow{3}{*}{\small \shortstack{ Network \\Distribution}} &  \multirow{3}{*}{\small \shortstack{Direct \\Effect}} &  \multirow{3}{*}{\small \shortstack{Population \\Size}} & \multicolumn{7}{c}{Exposure Mapping} \\
 \cline{4-10}
  &  &  & \multicolumn{3}{c}{(I)} &\multicolumn{3}{c}{(II)} & (III)  \\
  \cmidrule(lr){4-6}\cmidrule(lr){7-9}\cmidrule(lr){10-10}
  &  &  &  ${e_1}^{\dagger}$ & ${e_1(\small{T_{\beta}^{(2)}}})^{\dagger}$ & $e_2$ & ${e_1}^{\dagger}$ & ${e_1(\small{T_{\beta}^{(2)}}})^{\dagger}$ & $e_2$ & ${e_1(\small{T_{\beta}^{(2)}}})$  \\[0.1cm]
  \hline
 \multirow{4}{*}{\small \shortstack{Small world \\network}}& 2 & 500 & 0.041 & 0.044 & 1.000 & 0.034 & 0.051 & 1.000 & 0.907\\[0.1cm]
 & 2 & 1000 & 0.041 & 0.049 & 1.000 & 0.041 & 0.049 & 1.000 & 0.996\\[0.1cm]
 & 0 & 500 & 0.044 & 0.047 & 1.000 & 0.057 & 0.037 & 1.000 & 0.905\\[0.1cm]
 & 0 & 1000 & 0.043 & 0.060 & 1.000 & 0.046 & 0.045 & 1.000 & 0.997\\[0.1cm]
\hline
\multirow{4}{*}{\small \shortstack{Stochastic block \\model}}& 2 & 500 & 0.040 & 0.037 & 1.000 & 0.045 & 0.043 & 1.000 & 0.957\\[0.1cm]
& 2 & 1000 & 0.047 & 0.047 & 1.000 & 0.048 & 0.041 & 1.000 & 0.763\\[0.1cm]
 & 0 & 500 & 0.048 & 0.044 & 1.000 & 0.043 & 0.051 & 1.000 & 0.959\\[0.1cm]
 & 0 & 1000 & 0.049 & 0.048 & 1.000 & 0.046 & 0.045 & 1.000 & 0.767\\[0.1cm]
 \hline
\end{tabular}
}
\end{table}

\subsection{Power curve and functional misspecification} 
We validate that Procedure \ref{procedure estimated} controls Type I error at a given significance level of choice in Table \ref{table: oracle method} and Table \ref{table: est residual method}. The power properties displayed reflect an exposure-mapping misspecification when a spillover effect (at varying depths) is incorrectly not detected. This includes failing to account for first-order spillover effects and, similarly, for second-order spillover effects when they are present in the data-generating mechanism. Both represent exposure-mapping misspecification, leading to a high rejection rate and demonstrating the method's power properties. More specifically, the exposure-mapping misspecifications considered provide no information about any form of treatment effect (either first-order or second-order). In this section, we investigate the method's power properties when the functional form of the exposure mapping is misspecified. We consider the exposure mapping that takes as input an indicator of whether a treated neighbor exists to be the correctly specified exposure mapping in the null statement. The true data-generating process, that is, DGP~\hyperref[eq: outcome_prop]{I}, has exposure mapping given by the proportion of treated neighbors. We present the exposure mapping considered in the null formally below.
\begin{equation}
    \label{eq: indicator exp mapping}
    {e_3}_i = (Z_i, \mathcal{I}( \sum_{j \in \mathcal{N}(i)} Z_j > 0)).
\end{equation}
\begin{figure}[t]
    \centering
    \includegraphics[width=0.7\textwidth]{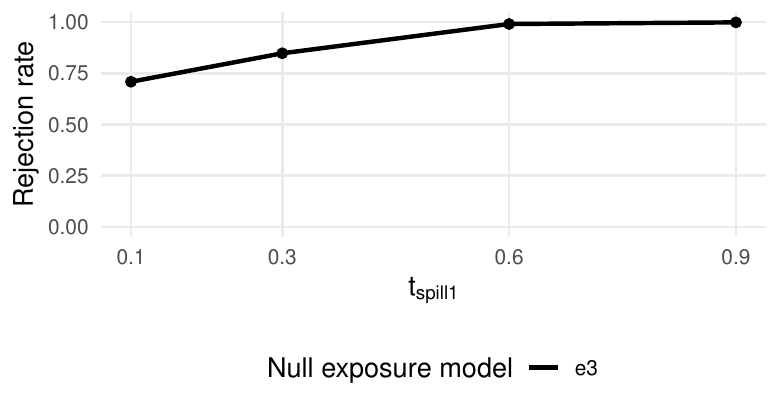}
    \caption{The Figure presents power properties of Procedure \ref{procedure estimated} under exposure mapping functional form misspecification. The rejection rates are plotted against the size of the first-order spillover effect as a power curve for the null exposure model.}
    \label{fig: power curve corr refit}
\end{figure}
We take $\tau_{direct} = 2$, and vary $\tau_{spill_1}$ between $0.1$ and $0.9$. We set the sample size to $500$ and obtain the power curve. The results are presented in Figure \ref{fig: power curve corr refit}. We implemented Procedure \ref{procedure estimated} (estimated residual randomization test) by sampling $500$ treatment assignment vectors to generate the null distribution. The graph in the data-generating process is coming from a small-world network, with the same parameters as before (number of initial connections, $K = 3$, and edge-rewiring probability, $p_{rw} = 0.1$). 

\subsection{Discussion}
Table \ref{table: oracle method} discusses results for the oracle procedure when true residuals are known. Table \ref{table: est residual method} presents results for the estimated residual method (Procedure \ref{procedure estimated}), where we estimate the residuals using linear regression. In both Tables \ref{table: oracle method} and \ref{table: est residual method}, Column $e_1$ for DGP~\hyperref[eq: outcome_prop]{I} and \hyperref[eq: outcome_prop_sq]{II} shows the Type I error of the method(s). DGP~\hyperref[eq: outcome_prop]{I} is a linear function of the exposure mapping $e_1$, and DGP~\hyperref[eq: outcome_prop_sq]{II} is a non-linear (quadratic) function of the exposure mapping $e_1$. We see that the Type I error is appropriately controlled at the chosen significance level of $0.05$ in the study. We note that no correction term was added to the procedure. When we test the procedure for $e_2$, which has only a direct-effect component, $Z_i$, and no spillover information, we observe high rejection rates. Thus, failure to account for the spillover effect (first order in the case of DGP~\hyperref[eq: outcome_prop]{I} and \hyperref[eq: outcome_prop_sq]{II}) is detected with high power. In the case of DGP~\hyperref[eq: outcome_prop_sq]{II}, we see similar power to DGP~\hyperref[eq: outcome_prop]{I}, despite the non-linearity and the use of standard linear regression for estimation. This again highlights how the method, without the addition of a corrective term, remains robust to Type I error inflation, even when some forms of model misspecification of the conditional response function are present. The last column test for exposure mapping $e_1$, which includes the proportion of treated neighbors as a component, was applied to the DGP~\hyperref[eq: outcome_prop_sec]{III}. DGP~\hyperref[eq: outcome_prop_sec]{III} contains both first-order and second-order spillover effects. Failing to account for the second-order spillover effect in exposure mapping constitutes model misspecification. We see our method performs well in this case as well, though it shows a drop in power compared to testing for the first-order spillover misspecification. This may be attributed to lower heterogeneity in the second-order spillover effect within the population. An interesting finding emerges when generating a graph from the Stochastic Block Model (SBM): a consistent drop in power as the sample size increases, more evident in Table \ref{table: oracle method}. We report that the average degree of an SBM instance generated in the study with 500 units is 3.6. When the sample size is increased to 1000 units, the average degree of a simulated SBM instance increases to 7.1. This renders a plausible reason that an increase in density reduces variability within the population, potentially leading to lower power, and also violates Assumption \ref{assumption: bdd degree}, which assumes a bounded degree of the graph. This is in contrast to the small-world network model, where the average degree stays the same (and equal to 3) with the increase in sample size (in line with Assumption \ref{assumption: bdd degree} of the maximum degree of the graph bounded by a constant), and we have an increasing power curve. We also consider a misspecification in the functional form of exposure mapping, as shown in Figure \ref{fig: power curve corr refit}. Here, both the hypothesized exposure mapping and the true data-generating mechanism are functions of first-order (treated) neighbors, but differ in their functional forms. We plot a power curve in Figure \ref{fig: power curve corr refit} where the null exposure mapping in consideration is a binary indicator of the presence of a treated neighbor. True exposure mapping is the same as that of DGP~\hyperref[eq: outcome_prop]{I}: the proportion of treated neighbors. The sample size is 500, and the graph is generated from the small-world network. We observe a strong power performance of the method under functional-form misspecification, which we plot for varying spillover effect sizes.\\
\begin{figure}[t]
    \centering
    \includegraphics[width=0.7\textwidth]{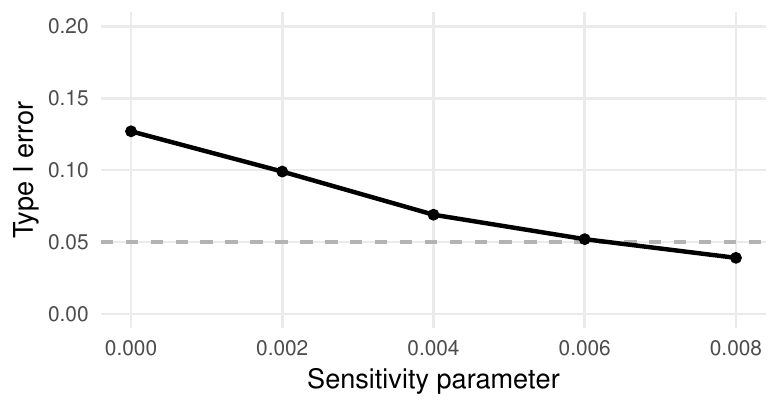}
    \caption{The Figure presents rejection rates of Procedure \ref{procedure estimated} for DGP~\hyperref[eq: outcome_prop_degree]{IV}. The rejection rates are plotted against the sensitivity parameter (scaled up to a constant), which is $C's$. Here, the sample size is 500, and the graph is generated from a small-world network model.}
    \label{fig: sensitivity plot}
\end{figure}
We now present a DGP in which we observe Type I error inflation due to model misspecification. We perform a grid search to empirically observe the Type I error rate for varying values of the sensitivity parameter, $s$, scaled by a constant in Figure \ref{fig: sensitivity plot}. Consider
\begin{equation}
\label{eq: outcome_prop_degree}
IV: Y_{i}(\mathbf{z}) =
\begin{cases} 
    {\tau_{\text{direct}}} \cdot z_{i} + (\tau_{\text{spill}_1}+\frac{deg_i}{N^{-1}\sum_j deg_j}) \cdot \frac{\sum_{j=1}^{N}G_{ij} \cdot z_{j}}{\sum_{j=1}^{N}G_{ij}} + \epsilon_{i} \qquad \sum_{j=1}^{N}G_{ij} > 0, \\
    \tau_{\text{direct}} \cdot z_{i} + \epsilon_{i} \qquad \qquad \qquad \qquad \qquad \quad \text{otherwise}.
\end{cases}
\end{equation}
Here, $deg_i= \sum_j G_{ij}$ is the degree of a unit $i$, and remaining parameters are same as DGP~\hyperref[eq: outcome_prop]{I}, with $\tau_{direct} = 2$, and $\tau_{spill_1} = 1.5$. We start with an uncorrected testing procedure and do a grid search over varying values of corrective factors (corresponding to $C's$) to obtain the sensitivity curve of Type I error inflation. $C' = 2C(1+2B'^2/\lambda)$ is a function of the Lipschitz constant, exposure mapping bound, and minimum eigenvalue of the design matrix, as shown in Appendix \ref{proof: estd asymptotic validity}. 

\section{Empirical Application}
\label{section: application}

We illustrate our methodology by re-analyzing the experimental dataset from \cite{paluck2016changing}, which examines the effect of anti-conflict norms on antagonistic behaviors among American middle school students. The authors set out to examine how peer influence can contribute to the spread of pro-social behavior to counter bullying, harassment, rumor-mongering, and social exclusion. In the experimental design, 28 of 56 schools were randomized to host the anti-conflict behavioral intervention program. About 3 weeks before the start of the experiment, all students at the 56 schools were asked to complete a survey. The questions would pertain to whom the student spends the most time with at school, to whom the student talks at school, etc., to map out the social network among school students. Refer to \cite{paluck2016changing} for further details on the experiment. Among the treated schools, the authors identified highly connected students in each school (about $10\%$ of the school's top nominations) and labeled them `social referents'. These social referents from the treated schools are, then again, randomized to treatment, with $50\%$ randomly invited to participate in an extensive bi-monthly meeting promoting anti-bullying and anti-conflict behavior. The non-social-referent students have a probability of zero of being in the treated group. The authors report the success of the anti-conflict intervention program by noting that students wore wristbands distributed as part of the program to encourage anti-conflict behaviors. \\
 
\begin{figure}[ht]
    \centering
     \hspace{-0.9cm}
    \includegraphics[width=0.45\textwidth, trim=4cm 2cm 4.1cm 2cm,
  clip]{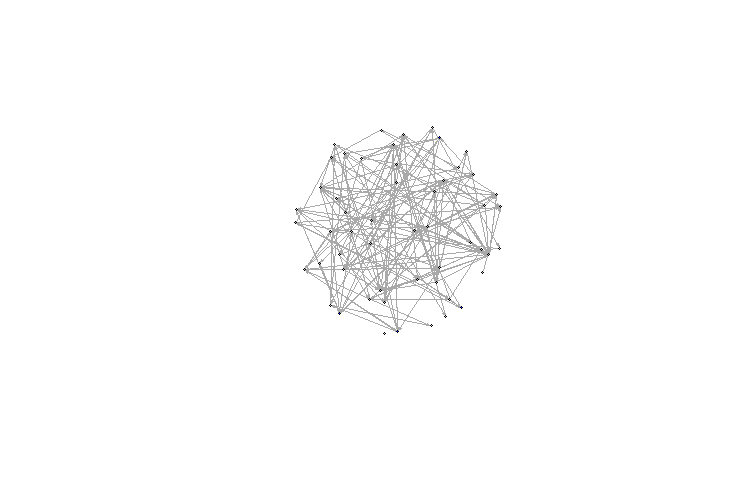}
    \hspace{0.4cm}
    \includegraphics[width=0.45\textwidth]{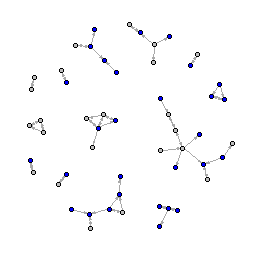}
    \caption{The above images are from one sample school in the data. The left image shows the full social network among the students, and the right image shows only the network restricted to social referents in the school. The blue units represent treated students, and the rest are in the control group.}
\end{figure}
Conditional on the first-stage randomization, we take the social referents from the treated schools as our population of interest. As part of our analysis, we removed all isolated nodes, resulting in a sample size of 850. We use a block-randomized trial to generate a null randomization distribution for testing model misspecification and include school fixed effects as a regressor to approximate a completely randomized trial. The average out-degree in the network for the sample is 1.14. 
\begin{figure}[htbp]
    \centering
    \includegraphics[width=0.8\textwidth]{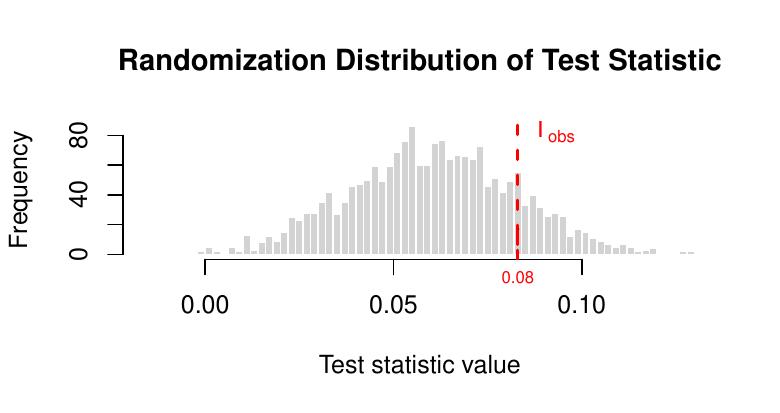}
    \caption{The plot is the randomization distribution of the test statistic, $T_{\beta}^{(1)}$. The null hypothesis considers $e = (Z_i, \mathcal{I}(\sum_{j \in \mathcal{N}(i)} Z_j > 0))$ to be correctly specified. We approximate the null distribution by taking 2,000 samples of the treatment assignment vector, $Z$, and plot the observed test statistic in Red.}
    \label{fig: test stat dist}
\end{figure}

\cite{paluck2016changing}, and subsequently re-analyzed by \cite{aronow2017estimating}, consider the exposure mapping to be $e_i= (Z_i, \mathcal{I}(\sum_{j \in \mathcal{N}(i)} Z_j > 0))$. Thus, the indication of whether there exists a treated friend in the school is assumed to capture the spillover mechanism in the population. We implement the specification test developed in the paper, namely Procedure \ref{procedure estimated} with the test statistic $T_{\beta}^{(1)}$, to test whether the exposure mapping $e$ is correctly specified. We generated the null distribution by taking 2,000 samples of the treatment assignment vector and obtained a p-value of 0.317. We plot the randomization distribution of the test statistic in Figure \ref{fig: test stat dist}. Since we fail to reject the null without applying a corrective term to account for model misspecification of the conditional linear response function, the decision is robust to such misspecification. That is, the addition of a corrective term to account for model-misspecification will increase the p-value estimate, which we also fail to reject. We also report p-values obtained for the null exposure mapping $e_2 = (Z)$, and for $e$, with the test statistic $T_{\beta}^{(2)}$. We obtain 0.055 and 0.799, respectively. A careful analysis of the selection of exposure mapping is beyond the scope of the paper, and we discuss sequential depth testing in Appendix \ref{sec: nested testing}.\\

\section{Conclusion}
\label{section: conclusion}
In this paper, we develop a randomization-based testing procedure to assess the correct specification of causal spillover mechanisms. The focus of the work has been on randomized experiments, in which the analyst knows the mechanism by which treatments are assigned. To this end, we build the testing framework and provide theoretical guarantees on the validity of the tests within the design framework. We use the linear projection of the conditional mean function to estimate the residuals. A key result of the paper is the asymptotic convergence of the OLS coefficients in a design sense if the exposure model is correctly specified. It is important to study the properties of other regression models from a design perspective, accounting for their dependence on the population. While the theory is built for linear basis model regression, we restrict our attention to the standard linear regression setting in the numerical studies. It will be an important extension to study, principally, how to choose the basis model and better calibrate the sensitivity parameter. Another interesting theoretical direction is to examine theoretical power properties of the procedure, and how departures of the conditional mean function from the assumed model space affect the performance of the randomization testing procedure. \cite{gao2026impossibility} presents an impossibility result for a general alternative model class, and it is an interesting extension to see if the impossibility result also holds if the alternative model space is restricted to a subclass, for instance, potential outcomes with a bounded degree of the network as considered in this work. Theoretical study of power properties will also aid in developing a principled procedure for selecting exposure models. Our procedure's performance drops in dense graph settings, creating a methodological gap to improve upon. A key contribution of the work has been the development of a framework for defining correctly specified causal inference models. Although this work focuses on interference settings, it would be interesting to investigate the proposed model-testing framework in other settings, including randomized experiments and observational studies.  

\vspace{0.25cm}
{\small
\textbf{Acknowledgments.} We thank Fabrizia Mealli, Fredrik Sävje, Qingyuan Zhao, JSM 2025 conference participants, and EuroCIM 2026 participants for their feedback. We gratefully acknowledge Hemanth Kumar at the ISB Institute of Data Science and, subsequently, Sarang Deo and the HPC Team at ISB for developing and providing access to the high-performance cluster.}

\bibliographystyle{plainnat}  
{\small \bibliography{references} }

\newpage
\begin{appendices}

\appendix  
\begin{center}
    {\Huge\appendixname}
\end{center}
\section{Proofs}
\subsection{Proof of Theorem \ref{theorem: oracle method validity}}
\label{proof: oracle method validity}
Consider $Z \sim P(Z)$. We define p-value as 
\begin{equation}
    pval = P(T(Z, R(Z))>T(Z_{obs}, R_{obs})). 
\end{equation}
Consider a variable $U$ with the same distribution as $T(Z, R(Z))$. Thus, $U$ is a random variable with the distribution of $U$ induced by $P(Z)$. This implies
\begin{equation}
    pval = 1- F_U(T(Z_{obs},R_{obs})). 
\end{equation}
Here, $F_U(.)$ is the cumulative distribution function of the variable U. That is,\quad $F(U \leq u) = P(U \leq u)$. \\
Since $R_i(Z) = R_i(Z_{obs})$ under the null $\mathcal{H}_0$, we get
\begin{equation}
\begin{aligned}
     T(Z, R(Z)) &= T(Z, R(Z_{obs}))\quad \forall Z \in \{0,1\}^N\: under \: \mathcal{H}_0,\\
     &\stackrel{d}{=} T(Z_{obs}, R_{obs}).
\end{aligned}
\end{equation}
Thus, the distribution of $U$ is the same as the distribution of $T(Z_{obs}, R_{obs})$. Here, the randomness is induced by $Z_{obs} \sim P(Z)$. Using this, we get
\begin{equation}
    pval = 1 - F_U(U).
\end{equation}
Using the probability integral transform theorem, we conclude
\begin{equation}
    P(pval \leq \alpha) \leq \alpha, 
\end{equation}
for a given $\alpha \in (0,1)$. 
\qed

\subsection{Proof of Lemma \ref{lemma: MSE minimizer}}
\label{proof: MSE minimizer}
Consider $I \sim Unif\{1,N\}$. Then
\begin{equation}
    \begin{aligned}
    L_{\mathbb{P}_N,Z} &= min_{m \in \mathcal{M}^N}  E_{I,Z}[(Y_I(Z) - m_I(e_I))^2],\\
         &= min_{m \in \mathcal{M}^N} E_Z[E_{I}[(Y_I(Z)-m_I(e_I))^2]],\\
        &= min_{m \in \mathcal{M}^N} E_Z[\frac{1}{N}\cdot \sum_i(Y_i(z)-m_i(e_i))^2|Z=z],\\
        &= min_{m \in \mathcal{M}^N} \frac{1}{N}\cdot \sum_i E_Z[(Y_i(z)-m_i(e_i))^2|Z=z].
    \end{aligned}
\end{equation}
Consider $L_i(z)= E_Z[(Y_i(z)-m_i(e_i))^2|Z=z]$. Then, 
\begin{equation}
    \begin{aligned}
        L_i(z) &= E_Z[((Y_i(z) - E_Z[Y_i|e_i]) + (E_Z[Y_i|e_i] - m_i(e_i)))^2 | Z=z], \\
        &= E_Z[(Y_i(z)-E_Z[Y_i|e_i])^2 + (E_Z[Y_i|e_i] - m_i(e_i))^2 \\
        &\qquad\qquad+ 2(Y_i(z)-E_Z[Y_i|e_i])\cdot(E_Z[Y_i|e_i] - m_i(e_i))|Z=z],\\
        &= E_Z[(Y_i(z)-E_Z[Y_i|e_i])^2|z] + E_Z[(E_Z[Y_i|e_i] - m_i(e_i))^2|Z=z] \\
        &\qquad\qquad + 2E_Z[(Y_i(z)-E_Z[Y_i|e_i])\cdot(E_Z[Y_i|e_i] - m_i(e_i))|Z=z]. 
    \end{aligned}
\end{equation} 
Consider the cross-over term we define as $C_i = 2E_Z[(Y_i(z)-E_Z[Y_i|e_i])\cdot(E_Z[Y_i|e_i] - m_i(e_i))|Z]$. Then, 
\begin{equation}
\begin{aligned}
    C_i &=2E_Z[(Y_i(z)-E_Z[Y_i|e_i])\cdot(E_Z[Y_i|e_i] - m_i(e_i))|Z=z], \\
        &= 2E_e[E_{Z|e}[(Y_i(z)-E_Z[Y_i|e_i]|e_i)\cdot(E_Z[Y_i|e_i] - m_i(e_i)|e_i)]],\\
        &= 2E_e[(E_Z[Y_i|e_i] - m_i(e_i)|e_i)\cdot E_{Z|e}[(Y_i(z)-E_Z[Y_i|e_i])]],\\
        &= 2E_e[(E_Z[Y_i|e_i] - m_i(e_i)|e_i)\cdot (E_Z[Y_i|e_i]-E_Z[Y_i|e_i])],\\
        &= 2E_e[(E_Z[Y_i|e_i] - m_i(e_i)|e_i)\cdot 0],\\
        &= 0. 
\end{aligned}
\end{equation}
Thus, 
\begin{equation}
    \begin{aligned}
         L_i(z) &=  E_Z[(Y_i(z)-E_Z[Y_i|e_i])^2|z] + E_Z[(E_Z[Y_i|e_i] - m_i(e_i))^2|Z=z].
    \end{aligned}
\end{equation}
Therefore, $argmin_m\; L_i(z) = argmin_m\; E_Z[(E_Z[Y_i|e_i] - m_i(e_i))^2|Z=z]$. If $\exists\: m_i \in \mathcal{M}_i$ s.t. $E_Z[Y_i|e_i] = m_i(e_i)$, then $argmin_m L_i(z) = E_Z[Y_i|e_i]$ a.s.
\qed

\subsection{Proof of Lemma \ref{lemma: consistency}}
\label{proof: consistency}
Define a random variable $E = e(Z)$. Consider $I \sim Unif\{1,N\}$. Then, under the null $\mathcal{H}_0$ 
\begin{equation}
    \begin{aligned}
        L_{\mathbb{P}_N,Z}(m) &= E_{I,Z}[(Y_I(z) - m(e_I))^2], \\
        &= \frac{1}{N}\sum_i\sum_z(Y_i(z)-m(e_i(z))^2\cdot P(Z=z),\\
        &= \frac{1}{N}\sum_i\sum_z(Y_i(e_i(z))-m(e_i(z))^2\cdot P(Z=z),\\
        &= \frac{1}{N}\sum_i\sum_{e_i}(Y_i(e_i)-m(e_i))^2\cdot P(E_i=e_i).\\
    \end{aligned}
\end{equation}
Take $e_{obs} \sim P(E)$. We note that this is induced by the probability distribution of $Z$. 
\begin{equation}
    \begin{aligned}
        \hat{L}_{e_{obs}}(m) &= \frac{1}{N}\cdot \sum_i(Y_i(e_{{obs}_i}) - m(e_{{obs}_i}))^2,\\
        &= \frac{1}{N}\cdot \sum_i\sum_{e_i}(Y_i(e_i) - m(e_i))^2\cdot I(e_{{obs}_i} = e_i). 
    \end{aligned}
\end{equation}
Consider $e_{obs}$ as a random variable, with distribution $P(E)$.  
\begin{equation}
\begin{aligned}
    E[\hat{L}_{e_{obs}}(m)] &= \frac{1}{N}\cdot \sum_i\sum_{e_i}(Y_i(e_i) - m(e_i))^2\cdot P(e_{{obs}_i} = e_i),\\
    &= \frac{1}{N}\cdot \sum_i\sum_{e_i}(Y_i(e_i) - m(e_i))^2\cdot P(E_i = e_i),\\
    &= L_{\mathbb{P}_N,Z}(m).
\end{aligned}
\end{equation}
Hence, $\hat{L}_{e_{obs}}(m)$ is an unbiased estimator of $L_{\mathbb{P}_N,Z}(m)$. We define $\hat{L}_i(m) = (Y_i(e_{{obs}_i}) - m(e_{{obs}_i}))^2$. Since $m \in \mathcal{M}_e$ such that $\exists\: \bar{\beta}: m = \bar{\beta}\cdot [1, \phi(e)]$, we get \\
$\hat{L}_i(m) = \hat{L}_i(\bar{\beta})= (Y_i(e_{{obs}_i}) - \bar{\beta} \cdot ([1, \phi(e_{{obs}_i})]))^2$. \\
We have $|Y_i| \leq B$, and $e_G(Z) = e_G(Z_i, Z_{N_i},..., Z_{{N_i}^k})$ we get
\begin{equation}
    ||(Z_i, Z_{N_i},..., Z_{{N_i}^k})|| \leq \sqrt{1+ \Delta + ... + \Delta^k}\leq\sqrt{1+\kappa +...+ \kappa^k},
\end{equation}
using Assumption \ref{assumption: bdd outcome}, \ref{assumption: bdd degree}. Take $B' = \sup_{z \in \{0,1\}^N:\: ||z|| \leq \sqrt{1+\kappa +...+ \kappa^k}} \;||[1,\phi(e(z)]||$. Then, $B'<\infty$ since $e(.)$ is bounded over a compact domain and $\phi(.)$ is continuous. This implies
\begin{equation}
    0 \leq \hat{L}_i \leq (B + ||\bar{\beta}||B')^2 = B'',
\end{equation}
given $\bar{\beta}$. \\
Therefore, using Theorem 2.1 from \cite{janson2004large}, we get that given $\epsilon >0$, 
\begin{equation}
\label{eq: janson ieq}
    P\big(\big|\frac{1}{N}\sum_i \hat{L}_i - E[\frac{1}{N}\sum_i \hat{L}_i]\big| > \epsilon\big) \leq 2 \exp\!\left(\frac{-2N\epsilon^2}{\chi^*(L_G)B''^2}\right).
\end{equation}
Here, $\chi^*(L_G)$ represents the fractional chromatic number of the dependency graph of $\{\hat{L}\}_i$, denoted by $L_G$, and $\chi(L_G)$ represents it's chromatic number. We know that 
\begin{equation}
\label{eq: chromatic}
    \chi^*(L_G) \leq \chi(L_G) \leq \Delta^{2k} + 1 \leq \kappa^{2k}+1.
\end{equation}
Using Equation (\ref{eq: chromatic}) in Equation (\ref{eq: janson ieq}), we get 
\begin{equation}
    P\big(\big|\frac{1}{N}\sum_i \hat{L}_i - E[\frac{1}{N}\sum_i \hat{L}_i]\big| > \epsilon\big) \leq 2 \exp\!\left(\frac{-2N\epsilon^2}{(\kappa^{2k}+1)B''^2}\right).
\end{equation}
Hence, 
\begin{align}
    \frac{1}{N}\sum_i \hat{L}_i - E\big[\frac{1}{N}\sum_i \hat{L}_i\big] & \xrightarrow{p} 0,\nonumber\\
    \implies \hat{L}_{e_{obs}}(\bar{\beta}) - L_{\mathbb{P}_N,Z}(\bar{\beta}) &\xrightarrow{p} 0. 
\end{align}
\qed

\subsection{Proof of Lemma \ref{lemma: argmin consistency}}
\label{proof: argmin consistency}
 Consider the OLS closed-form solution of the linear projection coefficient
 \begin{equation}
 \begin{aligned}
     &\bar{\beta}_0 = argmin_{\bar{\beta}}L_{\mathbb{P},N}(\bar{\beta}),\\
     \implies\\
     &[\frac{1}{N}\sum_{i=1}^{N}E_{Z}[1,\phi(e_i)][1,\phi(e_i)]^T]\cdot\bar{\beta}_0 = \frac{1}{N}\sum_{i=1}^{N}E_Z[[1,\phi(e_i)]Y_i],
 \end{aligned}    
 \end{equation}
 and the corresponding sample estimator
 \begin{equation}
     \begin{aligned}
         &\hat{\bar{\beta}}_{obs} = argmin_{\bar{\beta}}\hat{L}_{e_{obs}}(\bar{\beta}),\\
         \implies\\
         &[\frac{1}{N}\sum_{i=1}^{N}[1,\phi({e_{obs}}_i)][1,\phi({e_{obs}}_i)]^T]\cdot\hat{\bar{\beta}}_{obs} = \frac{1}{N}\sum_{i=1}^{N}[1,\phi({e_{obs}}_i)]{Y_{obs}}_i.
     \end{aligned}
 \end{equation}
 Define $\hat{\Sigma}_N = \frac{1}{N}\sum_{i=1}^{N}[1,\phi({e_{obs}}_i)][1,\phi({e_{obs}}_i)]^T$, and $\Sigma_N:=E_Z[\hat{\Sigma}_N] = \frac{1}{N}\sum_{i=1}^{N}E_{Z}[1,\phi(e_i)][1,\phi(e_i)]^T$. Also, define $\hat{u}_N = \frac{1}{N}\sum_{i=1}^{N}[1,\phi({e_{obs}}_i)]{Y_{obs}}_i$, and $u_N = E_Z[\hat{u}_N] = \frac{1}{N}\sum_{i=1}^{N}E_Z[[1,\phi(e_i)]Y_i]$.\\
 Hence, $\Sigma_N\bar{\beta_0} = u_N$, and $\hat{\Sigma}_N\hat{\bar{\beta}}_{obs} = \hat{u}_N$. Since ${\hat{\Sigma}}_{p,q} = 1/N\sum_{i=1}^{N}\phi({e_i}_p)\phi({e_i}_q)$ for some $p,q \in [d+1]$, taking ${e_i}_1 = 1$. Then, in a similar fashion as Lemma \ref{lemma: consistency}, we use Theorem 2.1 from \cite{janson2004large} , to obtain 
 \begin{equation}
 \label{eq: matrix conv}
     {\hat{\Sigma}}_{p,q} - {\Sigma}_{p,q} \xrightarrow{p} 0 \quad \forall\: {p,q}\in[d+1]. 
 \end{equation}
 This is using $|\phi({e_i}_p)\phi({e_i}_q)| \leq {B'}^2$, where $B'$ is as defined in Lemma \ref{lemma: consistency}, and the chromatic number of the dependency graph is upper bounded by $\kappa^{2k}+1$. Using Equation (\ref{eq: matrix conv}), we get
 \begin{equation}
 \begin{aligned}
     \label{eq: matrix op norm conv}
     ||\hat{\Sigma}_N - \Sigma_N||_{op} &\leq \sqrt{\sum_{p=1}^{d+1}\sum_{q=1}^{d+1}(\hat{\Sigma}_{p,q} - \Sigma_{p,q})^2},\\ 
     &\xrightarrow{p} 0,
 \end{aligned}    
 \end{equation}
 using the continuous mapping theorem. \\
 Under the null $\mathcal{H}_0$, we also get 
 \begin{equation}
     \label{eq: vector conv}
     \hat{u}_q - u_q \xrightarrow{p} 0 \quad\forall\:q\in[d+1],
 \end{equation}
 using $|\phi({e_i}_p)Y_i| \leq BB'$, and the same dependency graph bound as before. Consider
 \begin{equation}
 \begin{aligned}
     \hat{\Sigma}_N(\hat{\bar{\beta}}_{obs} - \bar{\beta}_0) &= \hat{u}_N - \hat{\Sigma}_N\bar{\beta}_0,\\
     &= \hat{u}_N - (\hat{\Sigma}_N - \Sigma_N + \Sigma_N)\bar{\beta}_0,\\
     &= \hat{u}_N - u_N + u_N - (\hat{\Sigma}_N - \Sigma_N)\bar{\beta}_0 - \Sigma_N\bar{\beta}_0,\\
     &= \label{eq: beta equality}(\hat{u}_N - u_N) - (\hat{\Sigma}_N - \Sigma_N)\bar{\beta}_0.   
 \end{aligned}    
 \end{equation}
 Under Assumption \ref{ass: design matrix pd}, we have there exists $\lambda > 0$, such that $\lambda_{min}(\Sigma_N) \geq \lambda$ for all $N$. Using Weyl's inequality, we know that for two symmetric matrices $S_1$ and $S_2$, 
 \begin{equation}
     \lambda_{min}(S_1+S_2) \geq \lambda_{min}(S_1) + \lambda_{min}(S_2). 
 \end{equation}
 Take $S_1 = \Sigma_N$, and $S_2 = \hat{\Sigma}_N - \Sigma_N$, and we get
 \begin{equation}
 \label{eq: lambda lower bdd}
 \begin{aligned}
     \lambda_{min}(\hat{\Sigma}_N) &\geq \lambda_{min}(\Sigma_N) + \lambda_{min}(\hat{\Sigma}_N-\Sigma_N),\\
     &\geq \lambda_{min}(\Sigma_N) - ||\hat{\Sigma}_N - \Sigma_N||_{op},
 \end{aligned}    
 \end{equation}
 using $|\lambda_{max}(.)| \leq ||.||_{op}$. \\
 Therefore, 
 \begin{equation}
     P(||\hat{\bar{\beta}}_{obs} - \bar{\beta}_0|| > \epsilon) = P(A)\cdot P(||\hat{\bar{\beta}}_{obs} - \bar{\beta}_0|| > \epsilon\:|\:A) + P(A^c)\cdot P(||\hat{\bar{\beta}}_{obs} - \bar{\beta}_0|| > \epsilon\:|\:A^c),
 \end{equation}
 where $A= \{||\hat{\Sigma}_N - \Sigma_N||_{op} \leq \lambda/2\}$. Hence, we get given $\epsilon > 0$,
 \begin{equation}
 \begin{aligned}
    P(||\hat{\bar{\beta}}_{obs} - \bar{\beta}_0|| > \epsilon) &\leq P(A)\cdot P(||\hat{\bar{\beta}}_{obs} - \bar{\beta}_0|| > \epsilon\:|\:A) + P(A^c), \\
    & \leq P(A)\cdot P(||{\hat{\Sigma}}^{-1}_N||_{op}(||\hat{u}_N -u_N|| + ||\hat{\Sigma}_N- \Sigma_N||_{op}||\bar{\beta}_0||) > \epsilon \:|\: A) + P(A^c),\\
    &\leq P(\frac{2}{\lambda}(||\hat{u}_N - u_N|| + \lambda^{-1}BB'||\hat{\Sigma}_N- \Sigma_N||_{op})>\epsilon\:|\:A) + P(A^c).
 \end{aligned}    
 \end{equation}
 Here, we use the following 
 \begin{equation}
  \begin{aligned}
  \label{eq: population beta bound}
      ||\bar{\beta}_0|| &= ||\Sigma_N^{-1}\cdot u_N||,\\
      &\leq ||\Sigma_N^{-1}||_{op}\cdot||u_N||,\\
      &\leq \lambda^{-1}BB'.
  \end{aligned}   
 \end{equation}
 Using Equation (\ref{eq: vector conv}), (\ref{eq: matrix op norm conv}), we get 
 \begin{equation}
     P(||\hat{\bar{\beta}}_{obs} - \bar{\beta}_0|| > \epsilon) \rightarrow 0. 
 \end{equation}
 as $N\rightarrow \infty$. 
\qed

\subsection{Proof of Theorem \ref{theorem: estd asymptotic validity}}
\label{proof: estd asymptotic validity}

With some abuse of notation, we write $\tilde{r}$ as $r$ for readability. Since $T(Z, \hat{r})$ is uniformly Lipschitz in $r$, we get that given $Z$ 
\begin{equation}
    \begin{aligned}
        |T(Z, \hat{r}) - T(Z, r_0)| \leq C/\sqrt{N} ||\hat{r} - r_0||,
    \end{aligned}
\end{equation}
 for some $C>0$.
 We define $\hat{\bar{\beta}}(Z)$ as the coefficient obtained from $Y(Z) \sim [1,\phi(e(Z))]$, and 
 \begin{equation}
     \hat{\bar{\beta}}^r(Z) = argmin_{\bar{\beta}} \frac{1}{N}\sum_{i=1}^{N}(Y_i(Z_{obs}) - \hat{\beta}_{obs}\cdot \phi(e_{i}(Z_{obs})) + \hat{\beta}_{obs}\cdot \phi(e_i(Z)) - \bar{\beta}\cdot[1,\phi(e_i(Z))])^2,
 \end{equation}
 as the coefficient obtained from refitting. Also, define $r_i = Y_i - E[\tilde{Y}_i|e_i]$, $\hat{r}_i= Y_i(Z) - \hat{\beta}(Z)\cdot\phi(e_i(Z))$, ${r_0}_i = Y_i(Z) - \beta_0\cdot\phi(e_i(Z))$, and $\delta_i = \beta_0\cdot\phi(e_i) - E[\tilde{Y}_i|e_i]$. We get $r_i = {r_0}_i + \delta_i$. We note that $r_i(Z) = r_i(Z')$ under the null $\mathcal{H}_0$. 
 Therefore, given $\epsilon > 0$, 
 \begin{equation}
 \begin{aligned}
     P(|T(Z,\hat{r}) - T(Z, r_0)| > \epsilon) &\leq P(C/\sqrt{N} \;||(\hat{r} - r_0)|| > \epsilon),\\
     &= P(||\hat{r} - r_0|| > \frac{\sqrt{N}\epsilon}{\small C }),\\
     &= P(||(\hat{\beta} - \beta_0)\cdot\phi(e(Z))|| > \frac{\sqrt{N}\epsilon}{\small C }),\\
     &\leq P(||(\hat{\beta} - \beta_0)|| >\frac{\epsilon}{CB'})\xrightarrow{} 0,      
\end{aligned}
 \end{equation}
   as $\hat{\beta} - \beta_0 \xrightarrow{p} 0$. Here, $B' = \sup_{z \in \{0,1\}^N:\: ||z|| \leq \sqrt{1+\kappa +...+ \kappa^k}} \;||[1,\phi(e_i(z))]||$ as defined in Lemma \ref{lemma: consistency}.\\
Hence, for two independent draws $Z, Z_{obs} \sim P(Z)$, we get 
 \begin{equation}
 \label{eq: z prob conv}
     T(Z, \hat{r}+\delta) - T(Z, r_0+\delta) \xrightarrow{p} 0,
 \end{equation}
 and 
 \begin{equation}
 \label{eq: z' prob conv}
     T(Z,\hat{r}_{obs}+\delta' + (\hat{\beta}_{obs} - \hat{\beta}^r(Z))\cdot\phi(e(Z))) - T(Z,r_{0obs}+\delta_{obs}) \xrightarrow{p_{Z,Z_{obs}}} 0,
 \end{equation}
 where $\delta' = \delta_{obs} + (\hat{\beta}^r(Z) - \hat{\beta}_{obs})\cdot\phi(e(Z))$. Let 
 \begin{equation}
     \Delta_{\hat{r}} = T(Z,\hat{r}+\delta) - T(Z_{obs},\hat{r}_{obs} +\delta_{obs}),
 \end{equation}
 and 
 \begin{equation}
     \Delta_{\hat{r}_{obs}} = T(Z,\hat{r}_{obs}+\delta' + (\hat{\beta}_{obs} - \hat{\beta}^r(Z))\cdot\phi(e(Z))) - T(Z_{obs},\hat{r}_{obs} + \delta_{obs}).
 \end{equation}
  Then, 
  \begin{equation}
   \label{eq: delta conv}   
   \Delta_{\hat{r}} - \Delta_{\hat{r}_{obs}} \xrightarrow{P_{Z,Z_{obs}}} 0.
  \end{equation} using Equation (\ref{eq: z prob conv}), (\ref{eq: z' prob conv}), and using $r(z) = r(z_{obs})$. 
 Consider the following inequality, given $\epsilon >0$
 \begin{equation}
     I(\Delta_{\hat{r}} \geq 0) \leq I(\Delta_{\hat{r}_{obs}} \geq -\epsilon) + I(|\Delta_{\hat{r}} - \Delta_{\hat{r}_{obs}}| \geq \epsilon). 
 \end{equation}
 We get
 \begin{equation}
 \label{eq: pval ineq}
     P_Z(\Delta_{\hat{r}} \geq 0) \leq P_Z(\Delta_{\hat{r}_{obs}} \geq -\epsilon) + P_Z(|\Delta_{\hat{r}} - \Delta_{\hat{r}_{obs}}| \geq \epsilon).
 \end{equation}
 We define ${\hat{p}}^{\epsilon}_{obs} = P_Z(\Delta_{\hat{r}_{obs}} \geq -\epsilon)$ and $\hat{p}_{0} = P_Z(\Delta_{\hat{r}} \geq 0)$. Consider
 \begin{align}
     &\{\hat{p}_{0} > \alpha + \epsilon\} \subseteq \{{\hat{p}}^{\epsilon}_{obs} >\alpha \} \cup \{P_Z(|\Delta_{\hat{r}} - \Delta_{\hat{r}_{obs}}| \geq \epsilon) > \epsilon\}.
 \end{align}
 Hence, we get
 \begin{align}
     P_{Z_{obs}}(\hat{p}_{0} > \alpha + \epsilon) &\leq P_{Z_{obs}}({\hat{p}}^{\epsilon}_{obs} > \alpha ) + P_{Z_{obs}}(P_Z(|\Delta_{\hat{r}} - \Delta_{\hat{r}_{obs}}| \geq \epsilon) > \epsilon),\\
      -P_{Z_{obs}}({\hat{p}}^{\epsilon}_{obs} > \alpha ) &\leq -P_{Z_{obs}}(\hat{p}_{0} > \alpha +\epsilon) + P_{Z_{obs}}(P_Z(|\Delta_{\hat{r}} - \Delta_{\hat{r}_{obs}}| \geq \epsilon) > \epsilon),\\
      1-P_{Z_{obs}}({\hat{p}}^{\epsilon}_{obs} > \alpha ) &\leq 1-P_{Z_{obs}}(\hat{p}_{0} > \alpha + \epsilon) + P_{Z_{obs}}(P_Z(|\Delta_{\hat{r}} - \Delta_{\hat{r}_{obs}}| \geq \epsilon) > \epsilon).
\end{align}
This implies
\begin{equation}
\label{eq: pval inequality}
P_{Z_{obs}}({\hat{p}}^{\epsilon}_{obs} \leq \alpha ) \leq P_{Z_{obs}}(\hat{p}_{0} \leq \alpha + \epsilon) + P_{Z_{obs}}(P_Z(|\Delta_{\hat{r}} - \Delta_{\hat{r}_{obs}}| \geq \epsilon) > \epsilon).
\end{equation}
 Consider $ P_{Z_{obs}}(P_Z(|\Delta_{\hat{r}} - \Delta_{\hat{r}_{obs}}| \geq \epsilon) > \epsilon)$. We first show that this term converges to 0 for any $\epsilon > 0$. \\
 Let $X = P_Z(|\Delta_{\hat{r}} - \Delta_{\hat{r}_{obs}}| \geq \epsilon)$. Since $X\geq 0$, using Markov inequality, we get
 \begin{align}
     P_{Z_{obs}}(X > \epsilon) &\leq \frac{E_{Z_{obs}}[X]}{\epsilon},\\
     &= \frac{E_{Z_{obs}}[E_Z[I(|\Delta_{\hat{r}} - \Delta_{\hat{r}_{obs}}| \geq \epsilon)]]}{\epsilon},\\
     &= \frac{E_{Z,Z_{obs}}[I(|\Delta_{\hat{r}} - \Delta_{\hat{r}_{obs}}| \geq \epsilon)]}{\epsilon},\\
     &= \frac{P_{Z,Z_{obs}}(|\Delta_{\hat{r}} - \Delta_{\hat{r}_{obs}}| \geq \epsilon)}{\epsilon},\\
     &\xrightarrow{} 0,
 \end{align}
 using Equation (\ref{eq: delta conv}), for any given $\epsilon > 0$.\\
 Since $T(Z,r) \stackrel{d}{=} T(Z_{obs},r_{obs})$ for $Z,Z_{obs} \sim_{i.i.d.} P(Z)$, we get $P_{Z_{obs}}(\hat{p}_0 \leq \alpha) \leq \alpha$ for all $N$, similar to the proof of Theorem \ref{theorem: oracle method validity}. In Equation (\ref{eq: pval inequality}), we obtain
 \begin{align}
 \label{eq: limsup ineq}
     &\limsup_{N \rightarrow \infty} P_{Z_{obs}}({\hat{p}}^{\epsilon}_{obs} \leq \alpha) \leq \limsup_{N \rightarrow \infty} P_{Z_{obs}}(\hat{p}_{0} \leq \alpha + \epsilon) \leq \alpha + \epsilon.
 \end{align}
 
 Hence, 
 \begin{align}
     \lim_{\epsilon \rightarrow 0}\limsup_{N \rightarrow \infty} P_{Z_{obs}}({\hat{p}}^{\epsilon}_{obs} \leq \alpha) \leq \alpha.
 \end{align}
  Consider $\hat{r}_{obs}^z = \hat{r}_{obs} + (\hat{\beta}_{obs} - \hat{\beta}^r(Z))\cdot\phi(e(Z))$, and 
   \begin{equation}
   \begin{aligned}
     \Delta_{\hat{r}_{obs}} = (T(Z,\hat{r}_{obs}^z+\delta') - T(Z,\hat{r}_{obs}^z)) + (T(Z,\hat{r}_{obs}^z)-T(Z_{obs},\hat{r}_{obs}))\\
     + (T(Z_{obs},\hat{r}_{obs}) - T(Z_{obs},\hat{r}_{obs} + \delta_{obs})).
     \end{aligned}
 \end{equation}
 Then, 
 $\Delta_{\hat{r}_{obs}} \leq C/\sqrt{N}(||\delta'||+||\delta_{obs}||) + (T(Z,\hat{r}_{obs}^z)-T(Z_{obs},\hat{r}_{obs}))$.
 This implies
 \begin{equation}
     \{\Delta_{\hat{r}_{obs}} \geq -\epsilon\} \subseteq \{(T(Z,\hat{r}_{obs}^z)-T(Z_{obs},\hat{r}_{obs})) \geq - C/\sqrt{N}(||\delta'||+||\delta_{obs}||)-\epsilon\}.
 \end{equation}
 Therefore, 
 \begin{equation}
 \begin{aligned}
 P_Z(\Delta_{\hat{r}_{obs}} \geq -\epsilon) &\leq P_Z(T(Z,\hat{r}_{obs}^z)-T(Z_{obs},\hat{r}_{obs}) \geq - C/\sqrt{N}(||\delta'||+||\delta_{obs}||)-\epsilon),
  \end{aligned}   
 \end{equation}
 which implies 
 \begin{equation}
     \{P_Z(T(Z,\hat{r}_{obs}^z)-T(Z_{obs},\hat{r}_{obs}) \geq - C/\sqrt{N}(||\delta'||+||\delta_{obs}||)-\epsilon) \leq \alpha\} \subseteq \{P_Z(\Delta_{\hat{r}_{obs}} \geq -\epsilon) \leq \alpha\},
 \end{equation}
 and we get, 
 \begin{equation}
 \begin{aligned}
    P_{Z_{obs}}(P_Z(T(Z,\hat{r}_{obs}^z)-T(Z_{obs},\hat{r}_{obs}) \geq - C/\sqrt{N}(||\delta'||+||\delta_{obs}||)-\epsilon) \leq \alpha) &\leq\\ P_{Z_{obs}}(P_Z(\Delta_{\hat{r}_{obs}} \geq -\epsilon) &\leq \alpha). 
 \end{aligned}
 \end{equation}
 Consider 
 \begin{equation}
 \begin{aligned}
     C/\sqrt{N}(||\delta'||+||\delta_{obs}||)) &\leq C/\sqrt{N}(2||\delta_{obs}|| + ||\hat{\beta}^r(Z) - \hat{\beta}_{obs}||\cdot||\phi(e(Z))||),\\
     &\leq C(2s+B'||\hat{\beta}^r(Z) - \hat{\beta}_{obs}||). 
 \end{aligned}    
 \end{equation}
 We now find a bound for the term $||\hat{\beta}^r(Z) - \hat{\beta}_{obs}||$. Similar to Lemma \ref{lemma: argmin consistency} excluding the baseline constants, take $\hat{\Sigma}_N = \frac{1}{N}\sum\phi(e_i)\phi(e_i)^T$, and $\Sigma_N = E_Z[\hat{\Sigma}_N]$. We have 
 \begin{equation}
 \begin{aligned}
     \hat{\Sigma}_N\cdot\hat{\beta}^r(Z) &= \frac{1}{N}\sum_i\phi(e_i)(\hat{\beta}_{obs}\cdot\phi(e_i) + \hat{r}_{{obs}_i}),\\
     &= \hat{\Sigma}_N\cdot\hat{\beta}_{obs} + \frac{1}{N}\sum_i\phi(e_i)\hat{r}_{{obs}_i}.
 \end{aligned}    
 \end{equation}
 Hence, 
 \begin{equation}
 \label{eq: beta bound main}
 \begin{aligned}
     \hat{\Sigma}_N\cdot(\hat{\beta}^r(Z)-\hat{\beta}_{obs}) &= \frac{1}{N}\sum_i\phi(e_i)({Y_{obs}}_i - \hat{\beta}_{obs}\cdot\phi({e_{obs}}_i)),\\
     &= \frac{1}{N}\sum_i\phi(e_i)({Y_{obs}}_i - \beta_0\cdot\phi({e_{obs}}_i) + (\beta_0 - \hat{\beta}_{obs})\cdot\phi({e_{obs}}_i)),\\
     &= \frac{1}{N}\sum_i\phi(e_i)(r_i - {\delta_{obs}}_i + (\beta_0 - \hat{\beta}_{obs})\cdot\phi({e_{obs}}_i)).
 \end{aligned}    
 \end{equation}
 Using arguments similar to proof of Lemma \ref{lemma: consistency}, we get under the null $\mathcal{H}_0$
 \begin{equation}
     \frac{1}{N}\sum_i\phi(e_i)r_i - \frac{1}{N}\sum_iE_Z[\phi(e_i)]r_i] \xrightarrow{p} 0,
 \end{equation}
 and the corresponding (population) first-order condition gives
 \begin{equation}
     \frac{1}{N}\sum_iE_Z[\phi(e_i)(Y_i - \beta_0\cdot\phi(e_i))] = 0, 
 \end{equation}
 implying 
 \begin{equation}
     \frac{1}{N}\sum_iE_Z[\phi(e_i)r_i] = \frac{1}{N}\sum_i E_Z[\phi(e_i)\delta_i].
 \end{equation}
 Using Jensen and Cauchy-Schwartz inequality, we have 
 \begin{equation}
  \label{eq: beta bound one}
  ||\frac{1}{N}\sum_iE_Z[\phi(e_i)\delta_i]|| \leq B's. 
 \end{equation}
 Similarly, we have 
 \begin{equation}
  \label{eq: beta bound two}
  ||\frac{1}{N}\sum_i\phi(e_i)\delta_{{obs}_i}|| \leq B's,
 \end{equation}
 and 
 \begin{equation}
 \label{eq: beta bound three}
 ||\frac{1}{N}\sum_i\phi(e_i)\phi(e_{{obs}_i})^T(\beta_0-\hat{\beta}_{obs})|| \leq B'^2||\hat{\beta}_{obs} - \beta_0|| \xrightarrow{p} 0,
 \end{equation}
 using Lemma \ref{lemma: argmin consistency}. We use Equations (\ref{eq: beta bound one}), (\ref{eq: beta bound two}), and (\ref{eq: beta bound three}) in Equation (\ref{eq: beta bound main}). We also use Equation (\ref{eq: lambda lower bdd}) to get $\lambda_{min}(\hat{\Sigma}_N) \geq \lambda/2$ w.p.1, and obtain under the null $\mathcal{H}_0$
 \begin{equation}
  \label{eq: beta bound final}
  ||\hat{\beta}^r(Z) - \hat{\beta}_{obs} ||\leq \frac{4B's}{\lambda} + o_p(1). 
 \end{equation}
 Therefore, 
 \begin{equation}
  \begin{aligned}
    \lim_{\epsilon\rightarrow0}\limsup_{N\rightarrow\infty}P_{Z_{obs}}(P_Z(T(Z,\hat{r}_{obs}^z)-T(Z_{obs},\hat{r}_{obs})& \geq - 2Cs(1+2B'^2/\lambda)-\epsilon) \leq \alpha) \\
    &\leq \lim_{\epsilon\rightarrow0}\limsup_{N\rightarrow\infty}P_{Z_{obs}}(\hat{p}^{\epsilon}_{obs} \leq \alpha),\\
    &\leq \alpha.
\end{aligned}
 \end{equation}
 Take $C' = 2C(1+2B'^2/\lambda)$. 
 \qed

\subsection{Proof of Proposition \ref{prop: test stat lip}}
\label{proof: lipschitz}

With some abuse of notation, we write $\tilde{r}$ as $r$ for readability. Let $r_{obs}(\beta) = Y_p - \bar{Y_p}$ where $Y_p = Y_{obs} - \beta\phi({e_{obs}})$. Similarly, let $r'_{obs}(\beta) = Y_p' - \bar{Y_p'}$ where $Y_p' = Y_{obs} - \beta'\phi(e_{obs})$. We first show that $1/N\cdot \sum_i(Y_{p_i} - \bar{Y_p})^2 \geq 1/N\cdot\sum_i({Y_{p}}_i(\beta_{obs}))^2$. Consider
\begin{equation}
\tilde Y := Y_{obs}-\bar Y_{obs}\mathbf 1,
\qquad
\tilde e := \phi(e_{obs})-\bar{\phi(e}_{obs})\mathbf 1,
\end{equation}
so that
\begin{equation}
Y_p-\bar Y_p\mathbf{1}
=
\tilde Y-\tilde e\beta.
\end{equation}
Hence,
\begin{equation}
\frac{1}{N}\sum_{i=1}^N\bigl(Y_{p_i}-\bar Y_p\bigr)^2
=
\frac{1}{N}\|\tilde Y-\tilde e\beta\|^2.
\end{equation}

Let
\begin{equation}
\beta_{obs}
=
(\tilde e^\top \tilde e)^{-1}\tilde e^\top \tilde Y,
\qquad
r_{obs}:=\tilde Y-\tilde e\beta_{obs}.
\end{equation}
Then,
\begin{equation}
\tilde e^\top r_{obs}=0.
\end{equation}

Now consider
\begin{equation}
\tilde Y-\tilde e\beta
=
r_{obs}-\tilde e(\beta-\beta_{obs}).
\end{equation}
And we get
\begin{equation}
\begin{aligned}
\|\tilde Y-\tilde e\beta\|^2
&=
\|r_{obs}-\tilde e(\beta-\beta_{obs})\|^2 \\
&=
\|r_{obs}\|^2
-2(\beta-\beta_{obs})^\top\tilde e^\top r_{obs}
+(\beta-\beta_{obs})^\top\tilde e^\top\tilde e(\beta-\beta_{obs}) \\
&=
\|r_{obs}\|^2
+(\beta-\beta_{obs})^\top\tilde e^\top\tilde e(\beta-\beta_{obs}).
\end{aligned}
\end{equation}
Since 
\begin{equation}
\frac{1}{N}\|r_{obs}\|^2
=
\frac{1}{N}\sum_{i=1}^N\bigl(Y_{p_{obs},i}-\bar Y_{p_{obs}}\bigr)^2,
\end{equation}
we have
\begin{equation}
\label{eq: var lower bdd}
\operatorname{Var}(Y_p)
=
\operatorname{Var}(Y_{p_{obs}})
+
\frac{1}{N}
(\beta-\beta_{obs})^\top
\tilde e^\top\tilde e
(\beta-\beta_{obs})
\;\ge\;
\operatorname{Var}(Y_{p_{obs}}).
\end{equation}

Let $r = Y(Z) - \beta\cdot\phi(e)$. Then, $r+\delta =_{\mathcal{H}_0} r_{obs}+\delta_{obs}$. Therefore, 
\begin{equation}
 \begin{aligned}
     ||r|| &= ||r_{obs} +\delta_{obs} - \delta||,\\
     &\geq ||r_{obs}|| - ||\delta - \delta_{obs}||,\\
     &\geq \sqrt{N}\:(\sqrt{c_0} - 2s).
 \end{aligned}   
\end{equation}
Consider $G^{(k)}$. We get that $||G^{(k)}||_{op} \leq \Delta^{k} \leq \kappa^{k}$. Then, 
\begin{equation}
\begin{aligned}
 \Big|T_{\beta}(z, r) - T_{\beta}(z, r')\Big| &= \Bigg| \frac{z^TG^{(k)}r}{||r||||z||} - \frac{z^TG^{(k)}r'}{||r'||||z||}\Bigg|,\\
 &= \Bigg|\frac{z^TG^{(k)}r}{||r||||z||} - \frac{z^TG^{(k)}r'}{||r'||||z||}\Bigg|,\\
 &= \Bigg|\frac{(z^TG^{(k)}r - z^TG^{(k)}r')}{||r||||z||} + \frac{z^TG^{(k)}r'}{||z||}\Big(\frac{1}{||r||}- \frac{1}{||r'||}\Big)\Bigg|.
 \end{aligned}
 \end{equation}
 Using the triangle inequality and the Cauchy-Schwarz inequality, we get
 \begin{equation}
 \begin{aligned}
 &\leq  \frac{|z^TG^{(k)}(r - r')|}{||r||||z||} + |z^TG^{(k)}r'|\Big|\frac{||r| - ||r'||}{||r||||r'||||z||}\Big|,\\
 &\leq \frac{||z||||G^{(k)}||_{op}||r - r'||}{||r||||z||} + ||z||||G^{(k)}||_{op}\Big(\frac{||r - r'||}{||r||||z||}\Big),\\
 &=  \frac{2||G^{(k)}||_{op}}{||r||} \Big(||r - r'||\Big).
 \end{aligned}
 \end{equation}
 Using the residual lower bound, we obtain 
 \begin{equation}
 \begin{aligned}
 \Big|T_{\beta}(z, r) - T_{\beta}(z, r')\Big| \leq \frac{2\kappa^k}{(\sqrt{c_0} - 2s)} \Big(\frac{||r - r'||}{\sqrt{N}}\Big).
  \end{aligned}
 \end{equation}
 
Take $C=\frac{2\kappa^k}{(\sqrt{c_0} - 2s)}$.  
\qed

\section{Additional results}
\label{sec: add results}
Without adjusting p-values, the probability of making at least one false-positive rejection increases. In this section, we present results on performing multiple model specification tests within the design-based experimental setup. We also present another example of a modeling assumption used in the literature to capture interference. As first proposed in \cite{cortez2023exploiting}, we keep the notation consistent and denote by $\beta$ the highest degree of the polynomial summand in the model. In this section, we denote the OLS coefficient with $\hat{b}$.

\subsection{Nested multiple specification testing}
\label{sec: nested testing}
Consider the multiple exposure misspecification testing setting described in Hypothesis \ref{hypothesis: multiple exposure}. Procedure \ref{procedure estimated} controls Type I error asymptotically at the chosen significance level of choice for a single randomization exposure specification test, as shown in Theorem \ref{theorem: estd asymptotic validity}. Consider implementing Procedure \ref{procedure estimated} individually for the Hypotheses defined in Hypothesis \ref{hypothesis: multiple exposure}. This naive approach to the multiple testing problem can lead to an inflated Family-Wise Error Rate (FWER). FWER is the probability of making at least one false rejection when performing multiple hypothesis testing at once. \cite{zhang2025multiple} presents results on the construction of nearly independent p-values under the randomization testing design framework. The independent p-values are then combined for testing intersection hypotheses using standard techniques such as Fisher's combination method, or Stouffer's method (\cite{Fishermultiple}; \cite{stouffer1949american}). A more general control of FWER control can be achieved through Hommel's method and other closed testing principle procedures (\cite{holm1979simple}; \cite{marcus1976closed}), though a formal investigation of these for randomization testing procedures is limited. \cite{zhong2024unconditional} considers multiple randomization testing for the depth of the spillover effect under interference, and obtains FWER control for finite samples (see Appendix E in \cite{zhong2024unconditional}). This is obtained through a sequential testing procedure that exploits a nested hypothesis-testing structure. In the same spirit, we consider settings where multiple causal model specification testing frameworks can be written in a nested structure. We, then, see that the FWER in such settings is controlled at the chosen significance level. We first define a nested hypothesis testing structure. 
\begin{definition}
\label{def: nested hyp}(Nested hypothesis testing)
 Consider $\mathcal{H}_{j}$ for $j \in [K]$ to be a set of hypotheses. We say $\{\mathcal{H}_j : j \in [K]\}$ is a nested hypothesis test if for $j \in [K]$
 \begin{equation}
     H_j \implies H_{j'} \quad \forall\;j' \geq j.
 \end{equation}
\end{definition}
For example, Hypothesis \ref{hypothesis: sequential test depth}, where testing for depth of spillover effect is a nested hypothesis test. In the subsequent section, we consider the low-order $\beta$ interactions model in Assumption \ref{assumption: low beta model}. Then, Hypothesis \ref{hypothesis: low beta} is an example of a nested hypothesis test. We develop a sequential testing procedure and show that the FWER is controlled appropriately. Consider
\begin{proced}
\label{procedure: sequential test}
 Consider $(Z_{obs}, Y_{obs})$ from an experimental data with known distribution $P(Z)$, and test statistics $T_j(Z, \hat{r})$ for the nested hypotheses $\{\mathcal{H}_j\}$ for $j \in [K]$.
 \begin{enumerate}
     \item Let $pval_j$ be the output from implementing Procedure \ref{procedure estimated} for the null Hypothesis $\mathcal{H}_j$ with $T_j$ as the test statistic.
     \item Take $j=1$.
     \item If $pval_j \leq \alpha$, reject $\mathcal{H}_j$ and continue to Step 4. Stop otherwise.
     \item Increment $j$ by 1, and repeat Step 3.
 \end{enumerate}
\end{proced}
The above sequential testing procedure controls the FWER by using a nested hypothesis-testing structure. The nested structure allows control of the FWER without adjusting the significance level, helping to retain the method's power. The spillover depth testing framework, as presented in Hypothesis \ref{hypothesis: sequential test depth}, qualifies as a nested hypothesis test in settings where the spillover effect is inversely proportional to geographical distance. Another example of a nested hypothesis test is considered in the subsequent section on testing for $\beta$, as shown in Hypothesis \ref{hypothesis: low beta}, assuming the low-order $\beta$ interactions model (Assumption \ref{assumption: low beta model}). If a lower-complexity model is correctly specified, it would imply that a higher-complexity model is also correctly specified, which falls under our nested testing approach. We now show that Procedure \ref{procedure: sequential test} controls FWER at the appropriate level. 
\begin{theorem}
 \label{thm: seq test FWER}
 Procedure \ref{procedure: sequential test} controls FWER at significance level $\alpha$ asymptotically. That is, given $\alpha$ and $j \in [K]$
 \begin{equation}
     \limsup_{N \rightarrow \infty} P_{Z_{obs}}(\cup_{j'=j}^{K} \cap_{m=1}^{j'}\{pval_{m} \leq \alpha\}| \mathcal{H}_{j'}: j \leq j'\leq K) \leq \alpha.
 \end{equation}
\end{theorem}
\begin{proof}
    Consider
    \begin{equation}
     \begin{aligned}
      \limsup_{N\rightarrow \infty}FWER &= \limsup_{N\rightarrow \infty}P_{Z_{obs}}(\cup_{j'=j}^{K} \cap_{m=1}^{j'}\{pval_{m} \leq \alpha\}| \mathcal{H}_{j'}: j \leq j'\leq K),\\
      &= \limsup_{N\rightarrow \infty}P_{Z_{obs}}(\cap_{m=1}^{j} \{pval_m \leq \alpha\}\:|\: \mathcal{H}_j: j \in [K]),\\
      &\leq \limsup_{N\rightarrow \infty}P_{Z_{obs}}(\{pval_j \leq \alpha\}\:|\: \mathcal{H}_j: j \in [K]),\\
      &\leq \alpha \qquad \text{using Theorem \ref{theorem: estd asymptotic validity}}.
     \end{aligned}   
    \end{equation}
\end{proof}

\subsubsection{Low $\beta$-order interactions model}
\label{section: example low beta order model}
Exposure mappings encode structural assumptions in the potential outcome model. This gives an abstract representation of the spillover mechanism, while keeping the setup relatively assumption-lean. However, it is not necessary for the exposure mappings' range to have controlled cardinality, provided sufficient dimensionality reduction. \cite{manski2013identification} suggests imposing additional functional assumptions to mitigate this dimensionality problem. \cite{cortez2023exploiting} presents a linear potential outcome model polynomial in treatment, in addition to the neighborhood interference assumption, which constrains the degree of the polynomial terms in the equation. Assuming a low-order interaction model can help build more efficient estimators and accommodate covariates or general experimental designs (e.g., \cite{wang2025covariate}, \cite{eichhorn2024low}). We present the low-order $\ beta$-interaction model below. 

\begin{assumption}(Neighbourhood interference)
\label{Assume: neighborhood interference}
Suppose 
\begin{equation}
    Y_i(Z) = Y_i(Z') \quad Z_{\bar{N_i}} = Z'_{\bar{N_i}}\: \forall\;Z,Z' \in \{0,1\}^N.
\end{equation}   
\end{assumption}
Assumption \ref{Assume: neighborhood interference} restricts the dependence of the potential outcome of a unit on the immediate neighborhood of the unit. With Assumption \ref{Assume: neighborhood interference}, \cite{cortez2023exploiting} proposes a linear model polynomial in terms of dependence on the treatment assignment of neighboring units. The degree of the polynomial is restricted to $\beta$, which has to be pre-fixed by the analyst. The summands comprise all combinatorial subsets of neighbors, with a maximum size of subset $\beta$.

\begin{assumption}(Low $\beta$-order interactions model)
\label{assumption: low beta model}
Consider
    \begin{equation}
        Y_i(Z) = \sum_{S \subseteq \bar{N_i}} c_{i,S}\cdot \prod_{j \in S}z_j. 
    \end{equation}
    Here, $c_{i,S} = 0$ for all $S$ such that $|S|>\beta$. 
\end{assumption}
By convention, we take $\prod_{j\in \phi}z_j = 1$. The subset restrictions run over $\bar{N_i}$ and thus include a linear direct effect in the model. Here, $c_{i,\phi}$ represents a baseline effect. While a lower $\beta$ in the model makes the estimators more efficient, the authors do not provide any criteria for investigating an appropriate $\beta$. Using the model specification procedure proposed in this paper, we provide a data-driven method to test for a value of $\beta$ rather than relying on ad hoc selection. To this end, we formulate a hypothesis-testing framework given the $\beta$ value. In the procedure, we assume that the potential outcome function follows the low-order interaction model with $\beta$. Thus, the functional form of the outcome model is assumed to be true. Misspecification error, if any, arises from a misspecified value of parameter $\beta$. The following framework formulates this problem as a hypothesis testing statement. 
\begin{hypothesis}
\label{hypothesis: low beta}($\beta$ is correctly specified)
    \begin{equation}
        \mathcal{H}_{\beta}: c_{i,S} = 0 \quad \forall\;S\subseteq \bar{N_i}:\; |S| > \beta\: \forall\;i\in[N].
    \end{equation}
    Here, $\beta \in \{1,2,...,\Delta\}$, where $\Delta$ is the maximum degree of the graph.
\end{hypothesis}
We construct a sequential test for $\beta$ within the nested hypothesis testing framework described above. A lower value of $\beta$ facilitates estimation and inference by limiting the number of parameters to be estimated, thereby reducing model dimensionality. However, if $\beta$ is as large as the maximum degree in the graph, we again encounter the curse of dimensionality. This is evident from the combinatorial growth of the number of coefficients $c_{i, S}$ in the model. Since $|S| \leq \beta$ and the coefficients are indexed by subsets $S$, the number of parameters grows at least as fast as $2^{\beta}$, with $\beta$ potentially as large as $|N_i|$. Consequently, having an appropriate $\beta$ is crucial: an underestimated $\beta$ induces model misspecification bias, while an overestimated $\beta$ imposes an unnecessary dimensionality burden. Implementing Procedure \ref{procedure estimated} to test if a given $\beta$ is correctly specified needs specification of a sufficient functional for the null statement.\\
We propose $e= [\Pi_{j \in S}z_j : S \subseteq \bar{N}_i, |S| \leq \beta]$ as the exposure mapping. We assume that the functional form of the low-order $\beta$ interactions model is correct and test the value of $\beta$. As the number of regressors can be much larger than the sample size, owing to combinatorial terms, a direct application of Procedure \ref{procedure estimated} is not feasible. It is an interesting extension to examine whether the sparsity arising from the graph's bounded degree (as described in Assumption \ref{assumption: bdd degree}) can be exploited to develop a model specification test based on high-dimensional regression.  

\end{appendices}
\end{document}